\documentclass{article}
\usepackage{ijcai22}

\usepackage{times}
\usepackage{soul}
\usepackage{url}
\usepackage[hidelinks]{hyperref}
\usepackage[utf8]{inputenc}
\usepackage[small]{caption}
\usepackage{graphicx}
\usepackage{amsmath}
\usepackage{amsthm}
\usepackage{amssymb}
\usepackage{booktabs}
\usepackage{algorithm}
\usepackage{algpseudocode}
\usepackage{multirow}
\usepackage[export]{adjustbox}
\usepackage[binary-units]{siunitx}

\usepackage{tikz}
\usetikzlibrary{calc}
\usetikzlibrary{positioning}
\usetikzlibrary{shapes.geometric}
\usetikzlibrary{patterns}
\tikzset{
vertex/.style={circle,draw,black,align=center,inner sep=0cm, minimum size=0.5cm,fill=white,anchor=center},
agent/.style={vertex,fill={rgb:red,1;green,1;blue,1},text=white},
agent1/.style={agent,fill={rgb:red,0;green,1;blue,3}},
agent2/.style={agent,fill={rgb:red,3;green,1;blue,0}},
agent3/.style={agent,fill={rgb:red,1;green,3;blue,2}},
line/.style={black}
}
\usepackage[font=small,subrefformat=parens]{subcaption}
\usepackage{macro}

\newtheorem{theorem}{Theorem}[section]
\newtheorem{lemma}[theorem]{Lemma}
\newtheorem{proposition}[theorem]{Proposition}

\newtheorem{definition}[theorem]{Definition}
\newenvironment{sketch}{\proof}{\endproof}
\newtheorem*{theorem*}{Theorem}
\newtheorem*{lemma*}{Lemma}
\newtheorem*{proposition*}{Proposition}

\usepackage{multibib}
\newcites{Appendix}{Additional References}

\title{Offline Time-Independent Multi-Agent Path Planning}

\author{
  Keisuke Okmura\footnote{Contact Author}\and
  Fran\c{c}ois Bonnet\and
  Yasumasa Tamura\And
  Xavier D\'{e}fago\\
  \affiliations
  Tokyo Institute of Technology\\
  \emails
  \{okumura.k, bonnet.f, tamura.y, defago.x\}@coord.c.titech.ac.jp
}

\begin{document}

\maketitle
\begin{abstract}
  This paper studies a novel planning problem for multiple agents that cannot share holding resources, named \emph{OTIMAPP (Offline Time-Independent Multi-Agent Path Planning)}.
  Given a graph and a set of start-goal pairs, the problem consists in assigning a path to each agent such that every agent eventually reaches their goal without blocking each other, regardless of how the agents are being scheduled at runtime.
  The motivation stems from the nature of distributed environments that agents take actions fully asynchronous and have no knowledge about those exact timings of other actors.
  We present solution conditions, computational complexity, solvers, and robotic applications.
\end{abstract}
\section{Introduction}
The eventual goal of collective path planning for multiple agents is to make each agent in a shared workspace be on their respective goal status.
This problem becomes non-trivial when agents cannot pass through each other, i.e., each agent occupies some resources in the space while the others are blocked to access these resources at that time.
We see such situations in fleet operations of warehouses~\cite{wurman2008coordinating}, intersection management for self-driving cars~\cite{dresner2008multiagent}, multi-robot 3D printing systems~\cite{zhang2018large}, packet-switched networks with limited buffer spaces~\cite{tel2000introduction}, and lock operations of transactions on distributed databases~\cite{knapp1987deadlock}, to name just a few.

In such multi-agent systems, each agent inherently takes and finishes actions (or moves) \emph{at their own timings independently and unpredictably from other actors}, regardless of centralized or decentralized controls.
This is due to the nature of \emph{distributed environments} such as message delay or clock shift/drift, as well as uncaptured individual differences between agents like frictions of physical robots.
Nevertheless, the cutting‐edge research on pathfinding for multiple agents, known as Multi-Agent Path Finding (MAPF)~\cite{stern2019def} that aims at finding a set of collision-free paths on graphs, heavily rely on timing assumptions.
Typical MAPF assumes that agents take actions just at the same time.
Not to mention, such ``timed'' schedules contradict the nature of distributed environments.
Even worse, on-time execution of offline planning is too optimistic with more agents.

One counter approach to the timing uncertainties is runtime supports by online monitoring, re-planning, and intervention, e.g.,~\cite{van2011reciprocal,ma2017multi,atzmon2020robust,okumura2021time}.
This approach however requires runtime effort and additional infrastructures (e.g., steady network and monitoring systems) to manage agents' status in real-time.
Moreover, how to realize such schemes in large systems is not trivial at all.

Instead, this paper studies a novel planning problem in which agents spontaneously take actions without any timing assumptions. The problem requests a set of paths (i.e., solution) ensuring that all agents eventually reach their destinations without blocking each other permanently.
To see this, consider the situation in Fig.~\ref{fig:example}(left).
This plan runs a risk of execution failure;
if the agent $j$ gets delayed for any reason while the agent $i$ moves two steps to the right, then each agent blocks each other and neither agent can progress on its respective path.
In contrast, in Fig.~\ref{fig:example}(right), regardless of how the two agents are scheduled, both agents eventually reach their destinations unless they permanently stop the progression.
We call the corresponding problem \emph{Offline Time-Independent Multi-Agent Path Planning (OTIMAPP)}.

{
  \newcommand{\edgesize}{0.2cm}
  \newcommand{\scaleratio}{1.0}
  \newcommand{\plotbody}{
    \node[agent1](v1) at (0.0, 0.0) {$i$};
    \node[vertex, right=\edgesize of v1](v2) {};
    \node[vertex, right=\edgesize of v2](v3) {};
    \draw[line,->,color={rgb:red,0;green,1;blue,3},very thick] (0.22, 0.15) -- (2.6, 0.15);
    \node[agent2, right=\edgesize of v3](v4) {$j$};
    \node[vertex, right=\edgesize of v4](v5) {};
    \node[vertex](v6) at (1.1, -0.7) {};
    \node[vertex, right=\edgesize of v6](v7) {};
    \foreach \u / \v in {v1/v2,v2/v3,v3/v4,v4/v5,v7/v4,v3/v6,v3/v7}
    \draw[line] (\u) -- (\v);
  }
  \begin{figure}[tb!]
    \centering
    \begin{tabular}{cc}
    \scalebox{\scaleratio}{
      \begin{minipage}{0.45\linewidth}
        \centering
        \begin{tikzpicture}
          \plotbody
          \draw[line,->,color={rgb:red,3;green,1;blue,0},very thick]
          (2.0, 0.05) -- (1.4, 0.05) -- (1.0, -0.6);
        \end{tikzpicture}
      \end{minipage}
      }
      &
      \scalebox{\scaleratio}{
        \begin{minipage}{0.45\linewidth}
          \centering
          \begin{tikzpicture}
            \plotbody
            \draw[line,->,color={rgb:red,3;green,1;blue,0},very thick]
            (2.1, -0.25) -- (1.9, -0.65) -- (1.5, -0.1) -- (1.15, -0.7);
          \end{tikzpicture}
        \end{minipage}
        }
    \end{tabular}
    \caption{
      Example of OTIMAPP.
      A graph is depicted with black lines.
      Two agents ($i$, $j$) and their paths are colored.
      \textit{left}: Both agents stop progression permanently due to mutual exclusion (i.e., no collision) if $i$ moved two steps before $j$ moves.
      \textit{right}: As long as each agent follows a respective path, both agents eventually reach their last vertex;
      these paths constitute an OTIMAPP solution.
    }
    \label{fig:example}
  \end{figure}
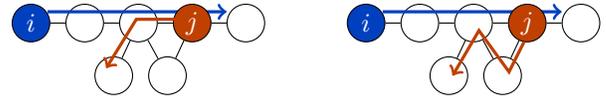
}

The contribution of this paper is to establish the foundation of OTIMAPP for both theory and practice.
Specifically, the topics are categorized into two:

\medskip
\noindent
$\blacktriangleright$~\emph{We formalize and analyze OTIMAPP.}
Section~\ref{sec:solution-analysis} identifies a necessary and sufficient condition for a \emph{solution}, i.e., a set of paths that makes all agents reach their goals without timing assumptions.
This is based on characterization of \emph{deadlocks}.
Section~\ref{sec:complexity} conducts a series of complexity analyses and reveals that
(1)~finding a solution is NP-hard on directed graphs,
(2)~finding a solution is NP-hard on undirected graphs when solutions are restricted to simple paths, and
(3)~verifying a solution is co-NP-complete.

\medskip
\noindent
$\blacktriangleright$~\emph{We present algorithms to solve OTIMAPP and demonstrate the utility of OTIMAPP via robotic applications.}
Section~\ref{sec:solvers} presents two approaches to derive solutions: prioritized planning (PP) and deadlock-based search (DBS).
Both algorithms are respectively derivative from basic MAPF algorithms~\cite{erdmann1987multiple,sharon2015conflict} and rely on a newly developed procedure to detect deadlocks within a set of paths.
Section~\ref{sec:eval} shows that either PP or DBS can compute large OTIMAPP instances to some extent.
Furthermore, we show that solutions keep robots' moves efficient in an adverse environment for timing assumptions compared to existing approaches with runtime supports~\cite{ma2017multi,okumura2021time}.
Moreover, we demonstrate that solutions are executable with physical robots in both a centralized style and a decentralized style with only local interactions, without cumbersome procedures of online interventions.

\medskip
In the remainder, all omitted proofs including sketches are available in the appendix.
The appendix, code, and movie are available on \url{https://kei18.github.io/otimapp}.
Related work will be discussed at the end.

\section{Problem Definition}
An \emph{OTIMAPP instance} is given by a graph $G = (V, E)$, a set of agents $A=\{1, 2, \ldots, N\}$, an injective initial state function $s:A \mapsto V$, and an injective goal state function $g:A \mapsto V$.
An OTIMAPP instance on digraphs is similar to the undirected case.

An \emph{execution schedule} is an infinite sequence of agents.
An \emph{OTIMAPP execution} is defined by an OTIMAPP instance, an execution schedule $\mathcal{E}$, and a set of paths $\{\path{1}, \ldots, \path{N}\}$ as follows.
The agents are \emph{activated} in turn according to $\mathcal{E}$.
Upon activation and until reaching the end of its path $\path{i}$, an agent $i$ takes a single step along $\path{i}$ if the vertex is vacant or stays at its current location otherwise.
After reaching the end of the path, the agent only stays.
$\mathcal{E}$ is called \emph{fair} when every agent appears infinitely-many times in $\mathcal{E}$.

An \emph{OTIMAPP problem} is to decide whether there is a set of paths $\{\path{1}, \ldots, \path{N}\}$ such that
(1)~each path for an agent $i$ begins from $s(i)$ and ends at $g(i)$,
(2)~for any fair execution schedule, all agents reach the end of their paths (i.e., goals) in a finite number of activations.
A \emph{solution} is a set of paths satisfying these two.

\paragraph{Other Notations}
Let $s_i$ and $g_i$ denote $s(i)$ and $g(i)$, respectively.
A location for an agent $i$ is associated with a \emph{progress index} $\clock_i \in \{1, \cdots, |\path{i}|\}$ and represented as $\loc{i}{\clock_i}$, where $\loc{i}{j}$ is the $j$-th vertex in \path{i}.
Every progress index starts at one and is incremented each time the agent moves a step along its path.
The progress index is non-decreasing and no longer increases after reaching the end of the path.
We use $S[-1]$ to denote the last element of the sequence $S$.

\paragraph{Rationale and remarks}
Any solution must deal with all timing uncertainties because execution schedules are unknown when offline planning.
We assume that agents are activated sequentially and that each activation is atomic.
However, there is no loss of generality as long as an agent can atomically reserve its destination before each move.
Indeed, several robots acted simultaneously in our demos.
Throughout the paper, we assume that each path \path{i} starts from $s_i$ and ends at $g_i$ to focus on analyses related to schedules.

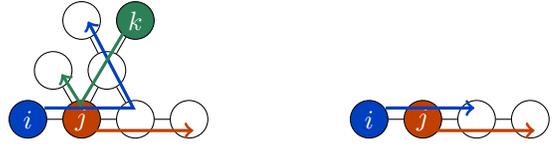
\begin{figure}[tb!]
  \centering
  \newcommand{\colwidth}{0.48\hsize}
  \newcommand{\edgesize}{0.2cm}
  \begin{tabular}{cc}
    \begin{minipage}[t]{\colwidth}
      \centering
      \begin{tikzpicture}
        \node[agent1](v1) at (0.0, 0.0) {$i$};
        \node[agent2,right=\edgesize of v1](v2) {$j$};
        \node[vertex,right=\edgesize of v2](v3) {};
        \node[vertex,right=\edgesize of v3](v4) {};
        \node[vertex](v5) at (1.05,0.65) {};
        \node[vertex,above=0.8cm of v2](v6) {};
        \node[agent3,above=0.8cm of v3](v7) {$k$};
        \node[vertex,left=\edgesize of v5](v8) {};
        \foreach \u / \v in {v1/v2,v2/v3,v3/v4,v2/v5,v3/v5,v5/v6,v5/v7,v2/v8}
        \draw[line] (\u) -- (\v);
        \draw[line,->,color={rgb:red,0;green,1;blue,3},very thick]
        (0.22, 0.15) -- (1.4, 0.15) -- (0.8, 1.3);
        \draw[line,->,color={rgb:red,3;green,1;blue,0},very thick] (0.9, -0.15) -- (2.2, -0.15);
        \draw[line,->,color={rgb:red,1;green,3;blue,2},very thick]
        (1.3, 1.2) -- (0.7, 0.2) -- (0.45, 0.6);
      \end{tikzpicture}
    \end{minipage}
    &
      \begin{minipage}[t]{\colwidth}
        \centering
        \begin{tikzpicture}
          \node[agent1](v1) at (0.0, 0.0) {$i$};
          \node[agent2,right=\edgesize of v1](v2) {$j$};
          \node[vertex,right=\edgesize of v2](v3) {};
          \node[vertex,right=\edgesize of v3](v4) {};
          \foreach \u / \v in {v1/v2,v2/v3,v3/v4}
          \draw[line] (\u) -- (\v);
          %
          \draw[line,->,color={rgb:red,0;green,1;blue,3},very thick] (0.22, 0.15) -- (1.4, 0.15);
          \draw[line,->,color={rgb:red,3;green,1;blue,0},very thick] (0.9, -0.15) -- (2.2, -0.15);
        \end{tikzpicture}
      \end{minipage}
  \end{tabular}
  \caption{
    Examples of unreachable potential deadlocks.
    \emph{left}: cyclic; $((i, j, k), (3, 1, 2))$.
    \emph{right}: terminal; $(i, j, 2)$.
  }
  \label{fig:counterexamples}
\end{figure}

\section{Solution Analysis}
\label{sec:solution-analysis}
Given a set of paths, our first question is to determine whether it is a solution.
This section derives a necessary and sufficient condition for solutions.
For this purpose, we introduce four types of \emph{deadlocks}, categorized as; \emph{cyclic} or \emph{terminal}; \emph{potential} or \emph{reachable}.
Informally, a cyclic deadlock is a situation where agent $i$ wants to move to the current vertex of $j$, who wants to move to the current vertex of $k$, who wants to move to ... of $i$.
A terminal deadlock is a situation where agent $i$ reaches its destination and blocks the progress of another agent $j$.
A potential deadlock is called reachable when there exists an execution schedule leading to the deadlock.

\begin{definition}[potential cyclic deadlock]
  Given an OTIMAPP instance and a set of paths $\{ \path{1}, \ldots \path{N} \}$, a \emph{potential cyclic deadlock} is a pair of tuples $\left((i, j, k, \ldots, l), (t_i, t_j, t_k, \ldots, t_l)\right)$ such that $\loc{i}{t_i+1} = \loc{j}{t_j} \land \loc{j}{t_j+1} = \loc{k}{t_k} \land \ldots \land \loc{l}{t_l+1} = \loc{i}{t_i}$.
  The elements of the first tuple are without duplicates.
  \label{def:potential-cycle-deadlock}
\end{definition}

\begin{definition}[potential terminal deadlock]
  Given an OTIMAPP instance and a set of paths $\{ \path{1}, \ldots \path{N} \}$, a \emph{potential terminal deadlock} is a tuple $(i, j, t_j)$ such that $\path{i}\left[-1\right] = \path{j}[t_j]$ and $i \neq j$.
  \label{def:potential-terminal-deadlock}
\end{definition}

\begin{definition}[reachable cyclic deadlock]
  A potential cyclic deadlock $\left((i, j, \ldots, l), (t_i, t_j, \ldots, t_l)\right)$ is \emph{reachable} when there is an execution schedule leading to a situation where $\clock_i = t_i \land \clock_j = t_j \land \ldots \land \clock_l = t_l$.
  This deadlock is called a \emph{reachable cyclic deadlock}.
  \label{def:reachable-cycle-deadlock}
\end{definition}

\begin{definition}[reachable terminal deadlock]
  A potential terminal deadlock $\left(i, j, t_j\right)$ is \emph{reachable} when there is an execution schedule leading to a situation where $\clock_i = |\path{i}| \land \clock_j = t_j - 1$.
  This deadlock is called a \emph{reachable terminal deadlock}.
  \label{def:reachable-terminal-deadlock}
\end{definition}

We refer to both reachable (or potential) cyclic/terminal deadlocks by reachable (resp. potential) deadlocks and simply use ``deadlock'' whenever the context is obvious.
At least one execution schedule is required to verify whether a potential deadlock is reachable.
For instance, in Fig.~\ref{fig:example} (left), a schedule $(i, i, \ldots )$ is evidence.
A potential deadlock is not always reachable as illustrated in Fig.~\ref{fig:counterexamples}.

\begin{theorem}[necessary and sufficient condition]
  Given an OTIMAPP instance, a set of path $\{ \path{1}, \ldots, \path{N} \}$ is a solution if and only if there are (1)~no reachable terminal deadlocks and (2)~no reachable cyclic deadlocks.
  \label{thrm:necessary-sufficient}
\end{theorem}
\begin{sketch}
  Verifying that they are necessary is straightforward.
  To see that they are sufficient, consider a potential function $\phi \defeq \sum_{i \in A} (|\path{i}| - \clock_i)$ defined over a configuration $\{ \clock_1, \ldots, \clock_N \}$.
  Observe that $\phi$ is non-increasing and $\phi=0$ means that all agents have reached their goals.
  Furthermore, when $\phi > 0$, $\phi$ is guaranteed to decrease if each agent is activated at least once.
\end{sketch}

\section{Computational Complexity}
\label{sec:complexity}
This section studies the complexity of OTIMAPP.
In particular, we address two questions: the difficulty to find solutions (Sec.~\ref{sec:complexity:finding}) and the difficulty to verify solutions (Sec.~\ref{sec:complexity:verification}).
Our main results are that both problems are computationally intractable;
the former is NP-hard and the latter is co-NP-complete.
Both proofs are based on reductions from the 3-SAT problem, deciding satisfiability for a formula in conjunctive normal form with three literals in each clause.

\subsection{Finding Solutions}
\label{sec:complexity:finding}
We distinguish directed graphs and undirected graphs to analyze the complexity.
The following proof is partially inspired by the NP-hardness of MAPF on digraphs~\cite{nebel2020computational}.

\begin{theorem}[complexity on digraphs]
  OTIMAPP on \emph{directed} graphs is NP-hard.
  \label{thrm:np-hard-directed}
\end{theorem}
\begin{proof}
  The proof is a reduction from the 3-SAT problem.
  Figure~\ref{fig:3-sat-finding} is an example of the reduction from a formula $(x_1 \lor x_2 \lor \lnot x_3) \land (\lnot x_1 \lor x_2 \lor x_3)$.

  \medskip
  \noindent
  \emph{A. Construction of an OTIMAPP instance.}
  We introduce two gadgets, called \emph{variable decider} and \emph{clause constrainer}.
  The OTIMAPP instance contains one variable decider for each variable and one clause constrainer for each clause.

  The variable decider for a variable $x_i$ assigns $x_i$ to true or false.
  This gadget contains one agent $\chi_i$ with two paths to reach its goal: \emph{left} or \emph{right}.
  Taking a left path corresponds to assigning $x_i$ to false, and vice versa.
  For the $j$-th clause $C^j$ in the formula, when its $k$-th literal is either $x_i$ or $\lnot x_i$, we further add one agent $c^j_k$ to the gadget.
  Its start and goal are positioned on the \emph{right} side from $\chi_i$ when the literal is a negation; otherwise, on the \emph{left} side.
  When several such agents are positioned on one side, let them connect (see the gadget for $x_2$).
  $c^j_k$ has two alternate paths to reach its goal: a path within the variable decider or a path via a clause constrainer.
  The former is available only when $\chi_i$ takes a path of the opposite direction to avoid a reachable cyclic deadlock.

  The clause constrainer for a clause $C^j$ connects the start and the goal of $c^j_k$.
  The gadget contains a triangle.
  Each literal $c^j_k$ enters this triangle from a distinct vertex and exits from another vertex.
  As a result, this gadget prevents three literals in $C^j$ from being false simultaneously; if not so, three agents enter the gadget and there is a reachable cyclic deadlock.

  The number of agents, vertices, and edges are all polynomial with respect to the size of the formula.

  \medskip
  \noindent
  \emph{B. The formula is  satisfiable if OTIMAPP has a solution}:
  the use of one clause constrainer by three agents leads to a reachable cyclic deadlock.
  Thus, at least one literal for each clause becomes true in any OTIMAPP solution.

  \medskip
  \noindent
  \emph{C. OTIMAPP has a solution if the formula is satisfiable}:
  If satisfiable, let $\chi_i$ take a path that follows the assignment.
  Let $c^j_k$ take a path within the variable decider when $\chi_i$ takes the opposite direction; otherwise, use the clause constrainer.
  Since three agents never enter one clause constrainer due to satisfiability, those paths constitute a solution.
\end{proof}

For undirected graphs, we limit solutions to those containing only simple paths.%
\footnote{
  We recently proved that it is NP-hard for the general case of undirected graphs.
  The formal proof will appear soon.
}

{
  \begin{figure}[t]
    \centering
    \includegraphics[width=1\linewidth]{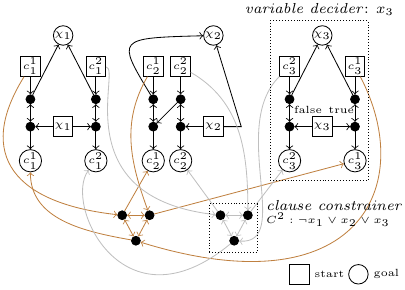}
    \caption{
      An OTIMAPP instance reduced from the 3-SAT formula $(x_1 \lor x_2 \lor \lnot x_3) \land (\lnot x_1 \lor x_2 \lor x_3)$.
    }
    \label{fig:3-sat-finding}
  \end{figure}
}

{
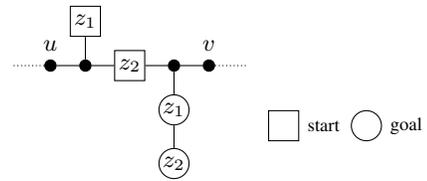
\begin{figure}
  \centering
  \begin{tikzpicture}
    \newcommand{\edgesize}{0.3cm}
    \newcommand{\edgesizehalf}{0.15cm}
    \newcommand{\edgesizethreetwo}{0.35cm}
    \newcommand{\longedgesize}{1.2cm}
    \small
    \tikzset{
      satnode/.style={vertex, minimum size=0.15cm, fill=black},
      start-node/.style={vertex, minimum size=0.4cm,rectangle},
      goal-node/.style={vertex, minimum size=0.4cm},
      line-c1/.style={brown,->},
      line-c2/.style={lightgray,->},
    }
    {
      \node[satnode, label=above:{\small $u$}](v1) at (0.0, 0.0) {};
      \node[satnode, right=\edgesize of v1](v2) {};
      \node[start-node, right=\edgesize of v2](v3) {$z_2$};
      \node[satnode, right=\edgesize of v3](v4) {};
      \node[satnode, right=\edgesize of v4,label=above:{\small $v$}](v5) {};
      \node[start-node, above=\edgesize of v2](v6) {$z_1$};
      \node[goal-node, below=\edgesize of v4](v7) {$z_1$};
      \node[goal-node, below=\edgesize of v7](v8) {$z_2$};
      \foreach \u / \v in {v1/v2,v2/v3,v3/v4,v4/v5,v6/v2,v4/v7,v7/v8}
      \draw[line] (\u) -- (\v);
      \draw[line, densely dotted] (v1) -- ($(v1) + (-0.5, 0.0)$);
      \draw[line, densely dotted] (v5) -- ($(v5) + (0.5, 0.0)$);
      \node[start-node, label=right:{\scriptsize start}](label-s) at (3.1, -0.8){};
      \node[goal-node, label=right:{\scriptsize goal}](label-g) at (4.2, -0.8){};
    }
  \end{tikzpicture}
  \caption{
    A \textit{oneway constrainer}.
    This gadget transforms an undirected edge $(u, v)$ to a directed one.
    Any agent is allowed to move only the way from $u$ to $v$ when limiting solutions to simple paths.
  }
  \label{fig:undirected-gadget}
\end{figure}
}

\begin{theorem}[complexity on undirected graphs]
  For OTIMAPP on \emph{undirected} graphs, it is NP-hard to find a solution with \emph{simple} paths.
  \label{thrm:np-hard-undirected}
\end{theorem}
\begin{sketch}
  We add a new gadget called \textit{oneway constrainer}, which transforms an undirected edge to a virtually directed one, to the proof of the NP-hardness on digraphs (Thm.~\ref{thrm:np-hard-directed}).
  We derive the claim by replacing all directed edges, except for bidirectional edges, with this gadget.
  Figure~\ref{fig:undirected-gadget} illustrates it, including two new agents: $z_1$ and $z_2$.
  In this gadget, any agents outside of the gadget are allowed to move only in the direction from $u$ to $v$.
\end{sketch}

\subsection{Verification}
\label{sec:complexity:verification}
The co-NP completeness of the verification relies on the following lemma, stating that finding cyclic deadlocks is computationally intractable.
Its entire proof is delivered in the Appendix.
\begin{lemma}[complexity of detecting cyclic deadlocks]
  Determining whether a set of paths contains either reachable or potential cyclic deadlocks is NP-complete.
  \label{lemma:deadlock-np-comp}
\end{lemma}
We then derive the complexity result since a solution has no reachable deadlocks.
\begin{theorem}[complexity of verification]
  Verifying a solutions of OTIMAPP is co-NP-complete.
  \label{thrm:co-np-hard:verification}
\end{theorem}
\begin{proof}
  Thm.~\ref{thrm:necessary-sufficient} states that a solution has no reachable terminal/cyclic deadlocks.
  Verifying no terminal deadlocks is in co-NP; a terminal deadlock is verified in polynomial time with an execution schedule.
  Verifying no potential deadlocks is co-NP-complete according to Lemma~\ref{lemma:deadlock-np-comp}.
\end{proof}

\section{Solvers}
\label{sec:solvers}
We now focus on how to solve OTIMAPP.
In practice, it is difficult to use the necessary and sufficient condition (Thm.~\ref{thrm:necessary-sufficient}) because we have to find corresponding schedules.
This motivates to build a relaxed sufficient condition.
\begin{theorem}[relaxed condition]
  Given an OTIMAPP instance, a set of path $\{ \path{1}, \ldots, \path{N} \}$ is a solution when there are
  (1)~no use of other goals, i.e., $g_j \not\in \path{i}$ for all $i \neq j$ except for $s_i = g_j$, and
  (2)~no potential cyclic deadlocks.
  \label{thrm:sufficient}
\end{theorem}
It is straightforward to see that the above conditions are respectively sufficient for the two conditions in Thm.~\ref{thrm:necessary-sufficient}.
Given a set of paths, ``no use of other goals'' is easy to check while ``no potential cyclic deadlocks'' is intractable to compute (Lemma~\ref{lemma:deadlock-np-comp}).
Nevertheless, \emph{detecting potential cyclic deadlock is the heart of solving OTIMAPP}.
Thus, we first explain how to detect potential cyclic deadlocks.
After that, two algorithms to solve OTIMAPP are presented.

\subsection{Detection of Potential Deadlocks}
Due to the space limit, we only describe the intuition behind the algorithm.
The details are in the Appendix (Alg.~\ref{algo:detecting-deadlock}).
We first introduce a \emph{fragment}, a candidate of potential cyclic deadlocks.
\begin{definition}[fragment]
  Given a set of paths $\{\path{1},$$\ldots,$$\path{N}\}$, a \emph{fragment} is a tuple $\left((i, j, k, \ldots, l), (t_i, t_j, t_k, \ldots, t_l)\right)$ such that $\loc{i}{t_i+1}=\loc{j}{t_j} \land \loc{j}{t_j+1} = \loc{k}{t_k} \land \ldots = \loc{l}{t_l}$.
  The elements of the first tuple are without duplicates.
  \label{def:fragment}
\end{definition}
\noindent
We say that a fragment \emph{starts} from a vertex $u$ when $\loc{i}{t_i} = u$ and a fragment \emph{ends} at a vertex $v$ when $\loc{l}{t_l+1} = v$.
A fragment that ends at its start (i.e., $\loc{l}{t_l+1} = \loc{i}{t_i}$) is a potential cyclic deadlock.

{
  \newcommand{\f}[3]{\m{\left[(#1), (#2), (#3)\right]}}
  \begin{table}[t]
    \centering
    {
    \footnotesize
    \begin{tabular}{rrl}
      \toprule
      induction & key & new fragments
      \\ \midrule
      $\{\pi_1\}$
          & $u$ & \f{1}{1}{u,v} \\
          & $v$ & \f{1}{2}{v,w}
      \\ \midrule
      $\{\pi_1, \pi_2\}$
          & $u$ & \f{1,2}{1,1}{u,v,x} \\
          & $v$ & \f{2}{1}{v,x} \\
          & $x$ & \f{2}{2}{x,y}
      \\ \midrule
      $\{\pi_{1}, \pi_{2}, \pi_{3}\}$
          & $u$ & \textcolor{blue}{\f{1,2,3}{1,1,2}{u,v,x,u}} \\
          & $v$ & \f{2,3}{1,2}{v,x,u} \\
          & $x$ & \f{3}{2}{x,u}, \f{3,1}{2,1}{x,u,v}\\
          & $z$ & \f{3}{1}{z,x}, \f{3,2}{1,2}{z,x,y}
      \\ \bottomrule
    \end{tabular}
    }
    \caption{
      Example of detecting potential cyclic deadlocks.
      We describe the update of \tablefrom for $\path{1} = (u, v, w), \path{2} = (v, x, y), \path{3} = (z, x, u)$.
      The table uses $[(\text{agents}), (\text{progress indexes}), (\text{path})]$ as a notation of fragment,  where ``path'' is a corresponding sequence of vertices of the fragment.
      The algorithm halts with a blue-colored fragment, a detected potential cyclic deadlock.
    }
    \label{table:update-tablefrom}
  \end{table}
}

Using fragments, we construct an algorithm to detect a potential cyclic deadlock in a set of paths if it exists.
This is based on induction on \path{i}.
The induction hypothesis for $i$ is that there are no potential cyclic deadlocks for $\{ \path{1}, \ldots, \path{i-1} \}$ and all fragments for them are identified.
All new fragments about \path{i} are categorized into three groups:
(1)~a fragment only with \path{i},
(2)~a fragment that extends existing fragments, and
(3)~a fragment that connects existing two fragments.
In either case, if a newly created fragment ends at its start, this is a deadlock.

The algorithm realizes this procedure by managing two tables that store fragments: \tablefrom and \tableto.
Both tables take one vertex as a key.
One entry in \tablefrom stores all fragments starting from the vertex.
One entry in \tableto stores all fragments ending at the vertex.
Table~\ref{table:update-tablefrom} presents an example to detect deadlocks.

\subsection{Prioritized Planning (PP)}
Prioritized planning~\cite{erdmann1987multiple,silver2005cooperative} is neither complete nor optimal, but it is computationally cheap hence a popular approach to MAPF.
It plans paths sequentially while avoiding collisions with previously planned paths.
Instead of inter-agent collisions, solvers for OTIMAPP have to care about potential cyclic deadlocks.

Algorithm~\ref{algo:pp} is prioritized planning for OTIMAPP, named \emph{PP}.
When planning a single-agent path, PP avoids using (1)~goals of other agents and (2)~edges causing potential cyclic deadlocks~[Line~\ref{algo:pp:planning}].
The latter is detected by storing all fragments created by previously computed paths.
For this purpose, PP uses the adaptive version of Alg.~\ref{algo:detecting-deadlock} [Line~\ref{algo:pp:register}] in the Appendix.
A path satisfying the constraints can be found by ordinary pathfinding algorithms.
If not, PP returns \FAILURE.
The correctness of PP is derived from Thm.~\ref{thrm:sufficient}.

{
  \begin{algorithm}[tb!]
  \caption{PP: Prioritized Planning}
  \label{algo:pp}
  \textbf{Input}:~an OTIMAPP instance\\
  \textbf{Output}:~a solution $\{ \path{1}, \ldots, \path{N} \}$ or \FAILURE
  \begin{algorithmic}[1]
    \State $\tablefrom, \tableto \leftarrow \emptyset$
    \For{$i = 1\ldots~|A|$}
    \State $\path{i} \leftarrow
    \begin{aligned}[t]
      &\text{a path for agent}~i~\text{while avoiding the use of}
      \\
      &\cdot~g_j, \forall j \neq i, \text{except for}~s_i
      \\
      &\cdot~(u, v) \in E~\text{s.t.}~\exists \theta \in \tableto[u]~\text{and}~\theta~\text{starts from}~v
      \\
      &\hspace{0.3cm}\triangleright\text{avoiding cyclic deadlocks for}~\path{j}, j < i
    \end{aligned}$
    \label{algo:pp:planning}
    \IFSINGLE{\path{i} is not found}{\Return \FAILURE}
    \State update \tablefrom and \tableto with \path{i} using Algorithm~\ref{algo:detecting-deadlock}
    \label{algo:pp:register}
    \EndFor
    \State \Return $\{ \path{1}, \ldots, \path{N} \}$
  \end{algorithmic}
\end{algorithm}
}

PP is simple but incomplete.
In particular, the planning order of agents is crucial; an instance may be solved or may not be solved as illustrated in Fig.~\ref{fig:prioritization}.
One resolution is repeating PP with random priorities until the problem is solved; let call this PP$^+$.
However, finding good orders can be challenging because there are $|A|!$ patterns.
This motivates us to develop a search-based solver, described in the next.

\begin{figure}[tb!]
  \centering
  \begin{tikzpicture}
    \newcommand{\edgesize}{0.6cm}
    %
    \node[agent1](v1) at (0.0, 0.0) {$i$};
    \node[vertex, right=\edgesize of v1.center](v2) {};
    \node[vertex, right=\edgesize of v2.center](v3) {};
    \node[vertex, right=\edgesize of v3.center](v4) {};
    \node[vertex](v5) at ($(v2) + (-0.5,  -0.7)$) {};
    \node[agent2](v6) at ($(v3) + ( 0.5,  -0.7)$) {$j$};
    \node[vertex](v7) at ($(v2) + ( 0.0,  -1.4)$) {};
    \node[vertex](v8) at ($(v3) + ( 0.0,  -1.4)$) {};
    \foreach \u / \v in {v1/v2,v2/v3,v3/v4,v2/v5,v3/v6,v5/v7,v6/v8,v7/v8}
    \draw[line] (\u) -- (\v);
    \draw[line,->,color={rgb:red,0;green,1;blue,3},very thick]
    (0.2, 0.2) -- (2.3, 0.2);
    \draw[line,->,color={rgb:red,3;green,1;blue,0},very thick, densely dotted]
    (2.0, -0.6) -- (1.7, -0.2) -- (0.9, -0.2) -- (0.6, -0.6);
    \draw[line,->,color={rgb:red,3;green,1;blue,0},very thick]
    (2.0, -0.8) -- (1.7, -1.2) -- (0.9, -1.2) -- (0.6, -0.8);
  \end{tikzpicture}
  \caption{
    Example that the planning order affects the solvability.
    When $i$ plans prior to $j$, PP results in success with solid lines.
    PP fails if $j$ plans first and takes the dotted line.
  }
  \label{fig:prioritization}
\end{figure}
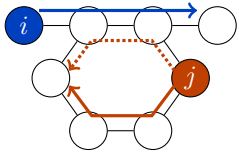

\subsection{Deadlock-based Search (DBS)}
We present \emph{deadlock-based search (DBS)} to solve OTIMAPP, based on a popular search-based MAPF solver called \emph{conflict-based search (CBS)}~\cite{sharon2015conflict}.
CBS uses a two-level search.
The high-level search manages collisions between agents.
When a collision occurs between two agents at some time and location, two possible resolutions are depending on which agent gets to use the location at that time.
Following this observation, CBS constructs a binary tree where each node includes constraints prohibiting to use space-time pairs for certain agents.
In the low-level search, agents find a single path constrained by the corresponding high-level node.

Instead of collisions, DBS considers potential cyclic deadlocks.
When detecting a deadlock in a set of paths, a resolution is that one of the agents in the deadlock avoids using the edge.
Thus, the constraints identify which agents prohibit using which edges in which orientation.

{
  \begin{algorithm}[tb!]
  \caption{DBS: Deadlock-based Search}
  \newcommand{\OPEN}{\m{\mathit{OPEN}}}
  \label{algo:cp}
  \textbf{Input}:~an OTIMAPP instance\\
  \textbf{Output}:~a solution $\{ \path{1}, \ldots, \path{N} \}$ or \FAILURE
  \begin{algorithmic}[1]
    \State $R.\mathit{constraints} \leftarrow \emptyset$
    \label{algo:cp:root:constraints}
    \State $R.\mathit{paths} \leftarrow$~find paths with ``no use of other goals''
    \label{algo:cp:root:paths}
    \State insert $R$ to \OPEN
    \label{algo:cp:root:insert}
    \Comment \OPEN: priority queue
    \While{$\OPEN \neq \emptyset$}
    \label{algo:cp:while}
    \State $N \leftarrow \OPEN.\mathit{pop}()$
    \label{algo:cp:pop}
    \State $C \leftarrow$ get constraints of $N$ using Algorithm~\ref{algo:detecting-deadlock}
    \label{algo:cp:get-constraints}
    \IFSINGLE{$C = \emptyset$}{\Return $N.\mathit{paths}$}
    \label{algo:cp:return-solution}
    \For{$(i, u, v) \in C$}
    \label{algo:cp:exapnd:for}
    \State $\p{N} \leftarrow
    \begin{aligned}[t] \{
      &\mathit{constraints}: N.\mathit{constraints} + (i, u, v),\\
      &\mathit{paths}: N.\mathit{paths}\}
    \end{aligned}$
    \State update \path{i} in $\p{N}.\mathit{paths}$ to follow $\p{N}.\mathit{constraints}$
    \label{algo:cp:update}
    \IFSINGLE{\path{i} is found}{insert $\p{N}$ to \OPEN}
    \EndFor
    \label{algo:cp:exapnd:for:end}
    \EndWhile
    \label{algo:cp:while:end}
    \State \Return \FAILURE
    \label{algo:cp:return-failure}
  \end{algorithmic}
\end{algorithm}
}

Algorithm~\ref{algo:cp} describes the high-level search of DBS.
Each node in the high-level search contains \emph{constraints}, a list of tuples consisting of one agent and two vertices (representing ``from vertex'' and ``to vertex''), and \emph{paths} as a solution candidate.
The root node does not have any constraints~[Line~\ref{algo:cp:root:constraints}].
Its paths are computed following ``no use of other goals'' of Thm.~\ref{thrm:sufficient}~[Line~\ref{algo:cp:root:paths}].
Then, the node is inserted into a priority queue \open~[Line~\ref{algo:cp:root:insert}].
In the main loop [Line~\ref{algo:cp:while}--\ref{algo:cp:while:end}], DBS repeats;
(1)~Picking up one node~[Line~\ref{algo:cp:pop}].
(2)~Checking a deadlock and creating constraints~[Line~\ref{algo:cp:get-constraints}].
(3)~Returning a solution if the paths contain no deadlocks~[Line~\ref{algo:cp:return-solution}].
(4)~If not, creating successors and inserting them to \open~[Line~\ref{algo:cp:exapnd:for}--\ref{algo:cp:exapnd:for:end}].
DBS returns \FAILURE when \open becomes empty~[Line~\ref{algo:cp:return-failure}].
We complement several details below.

{
\smallskip
\newcommand{\myitemize}[2]{\noindent{$\blacktriangleright$~\emph{#1}:}#2\smallskip}

\myitemize{Line~\ref{algo:cp:pop}}{
  \open is a priority queue and needs the order of nodes.
  DBS works in any order but good orders reduce the search effort.
  As effective heuristics, we use the descending order of the number of deadlocks with two agents, which is computed within a reasonable time.}

\myitemize{Line~\ref{algo:cp:get-constraints}}{
  Let $((i, j, k, \ldots, l),$$(t_i, t_j, t_k, \ldots, t_l))$ be a returned deadlock by Alg.~\ref{algo:detecting-deadlock}.
  Then, create constraints $(i, \loc{i}{t_i}, \loc{i}{t_i+1}), (j, \loc{j}{t_j}, \loc{j}{t_j+1}), \ldots, (l, \loc{l}{t_l}, \loc{l}{t_l+1})$.}

\myitemize{Line~\ref{algo:cp:update}}{
  forces one path \path{i} in the node to follow the new constraints.
  This low-level search must follow ``no use of other goals,'' furthermore, all edges in the constraints for $i$.
  If not found, DBS discards the corresponding successor.}
}

\begin{theorem}[DBS]
  DBS returns a solution when solutions satisfying Thm.~\ref{thrm:sufficient} exist; otherwise returns \FAILURE.
  \label{thrm:dbs}
\end{theorem}

\paragraph{Example}
We describe an example of DBS using Fig.~\ref{fig:prioritization}.
Assume that the initial path of $i$ is the solid blue line and the path for $j$ is the dashed red line [Line~\ref{algo:cp:root:paths}].
This node is inserted into \open~[Line~\ref{algo:cp:root:insert}] and is expanded immediately~[Line~\ref{algo:cp:pop}].
There is one potential cyclic deadlock in the paths then two constraints are created: either $i$ or $j$ avoids using the shared edge~[Line~\ref{algo:cp:update}].
Two child nodes are generated, however, the node that changes $i$'s path is invalid because there is no such path without the use of the goal of $j$.
Another one is valid; $j$ takes the solid red line.
Therefore, one node is added to \open from the root node.
In the next iteration, this newly added node is expanded.
There are no potential cyclic deadlocks in this node;
DBS returns its paths as a solution.

\paragraph{Optimization}
Although this paper focuses on a feasibility problem, DBS can adapt to optimization problems.
As objective functions, total path length and maximum path length in a solution can be defined.
Those optimization problems are solved optimally by DBS when it prioritizes high-level search nodes with smaller scores, as commonly done in CBS.
Note that metrics that assess time aspects such as total traveling time used in MAPF studies are significantly affected by execution schedules; the adaptation is not trivial.

\section{Evaluation}
\label{sec:eval}
This section empirically demonstrates that OTIMAPP solutions are computable to some extent (Sec.~\ref{sec:stress-test}) and they are useful in adverse environments about timings (Sec.~\ref{sec:delay-tolerance}) through the simulation experiments.
We also present OTIMAPP execution with robots (Sec.~\ref{sec:demo}).
The simulator was coded in C++ and the experiments were run on a desktop PC with Intel Core~i9 \SI{2.8}{\giga\hertz} CPU and \SI{64}{\giga\byte} RAM.

\subsection{Stress Test}
\label{sec:stress-test}
\paragraph{Setup}
Each solver was tested with a timeout of \SI{5}{\minute} on four-connected undirected grids picked up from~\cite{stern2019def}, as a graph $G$.
We also tested random graphs, shown in the Appendix.
All instances were generated by setting a start $s_i$ and a goal $g_i$ randomly while ensuring that $s_i$ and $g_i$ have at least one path without the use of other goals; otherwise, it violates ``no use of other goals'' of Thm.~\ref{thrm:sufficient}.
Note that unsolvable instances might still be included.

{
  \setlength{\tabcolsep}{0pt}
  \newcommand{\colwidth}{0.32\hsize}
  \newcommand{\plotStress}[1]{
    \begin{minipage}[t]{\colwidth}
      \centering
      \IfFileExists{fig/raw/stress-#1.pdf}{\includegraphics[width=1.0\linewidth]{fig/raw/stress-#1.pdf}}{}
    \end{minipage}
  }
  \newcommand{\plotCactus}[1]{
    \begin{minipage}[t]{\colwidth}
      \centering
      \IfFileExists{fig/raw/cactus-#1.pdf}{\includegraphics[width=1.0\linewidth]{fig/raw/cactus-#1.pdf}}{}
    \end{minipage}
  }
  \newcommand{\plotProfiling}[1]{
    \begin{minipage}[t]{\colwidth}
      \centering
      \IfFileExists{fig/raw/profiling-#1.pdf}
      {\includegraphics[width=0.95\linewidth,right]{fig/raw/profiling-#1.pdf}}{}
    \end{minipage}
  }
  \newcommand{\headerGrid}[4]{
    \begin{minipage}[t]{\colwidth}
      \centering
      {
        \setlength{\tabcolsep}{1pt}
        \renewcommand{\arraystretch}{0.8}
        \begin{tabular}{lr}
          {\scriptsize\field{#1}}
          & \multirow{2}{*}{\includegraphics[width=0.23\linewidth]{fig/raw/#1.pdf}}\\
          {\tiny\m{#2\stimes #3~(#4)}} &
        \end{tabular}
      }
    \end{minipage}
  }
  \begin{figure}[tb!]
    \centering
    \begin{tabular}[t]{ccc}
      \headerGrid{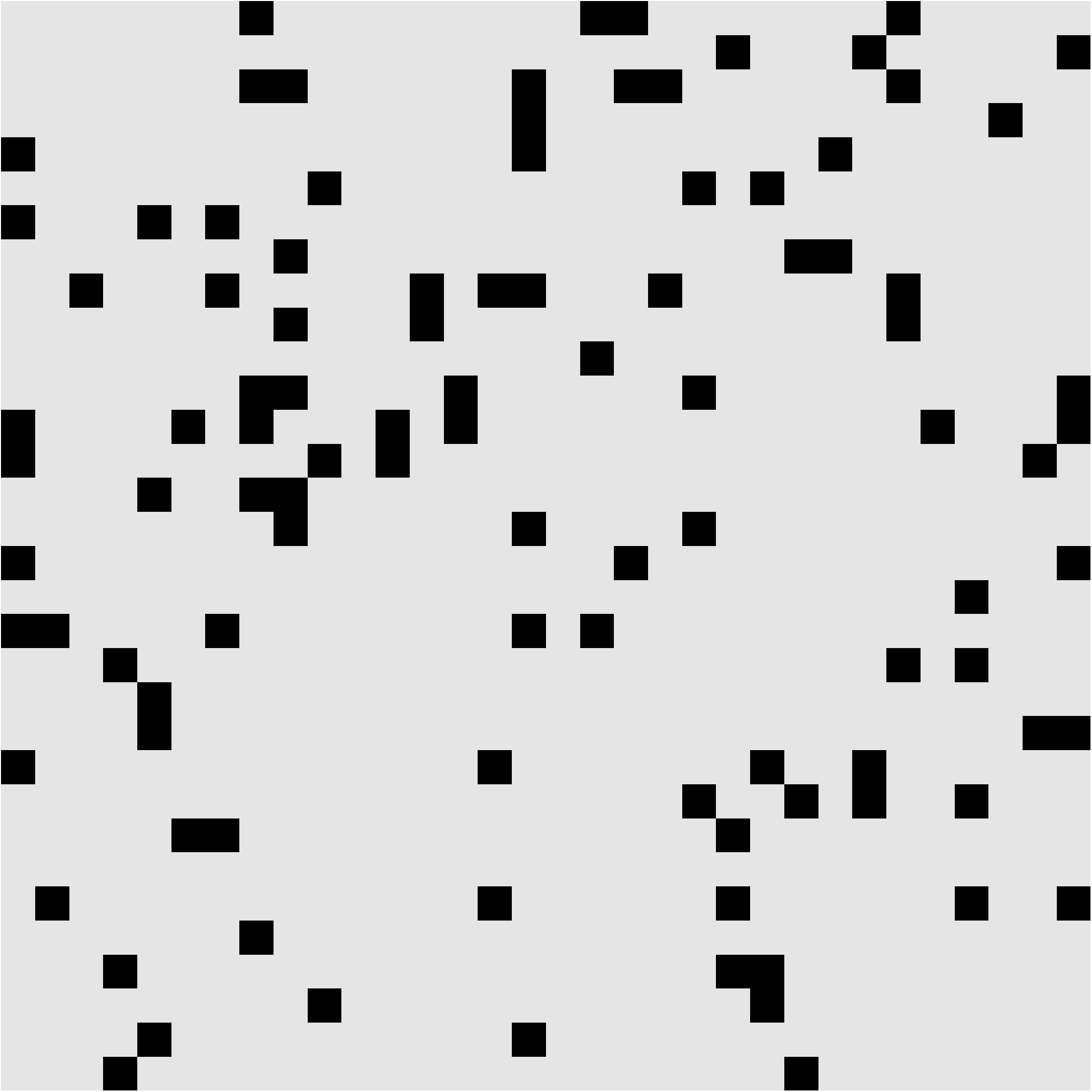}{32}{32}{922}
      & \headerGrid{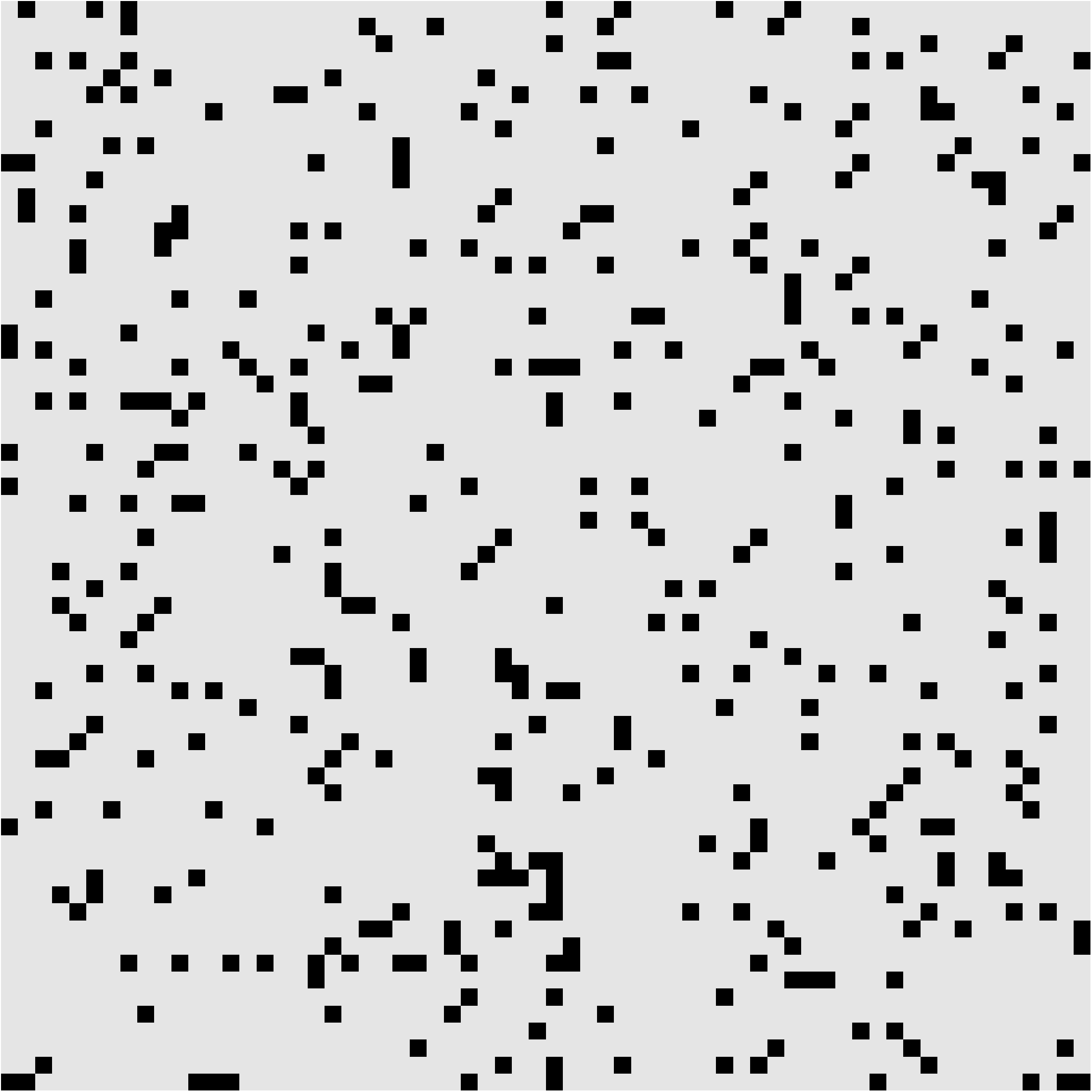}{64}{64}{3687}
      & \headerGrid{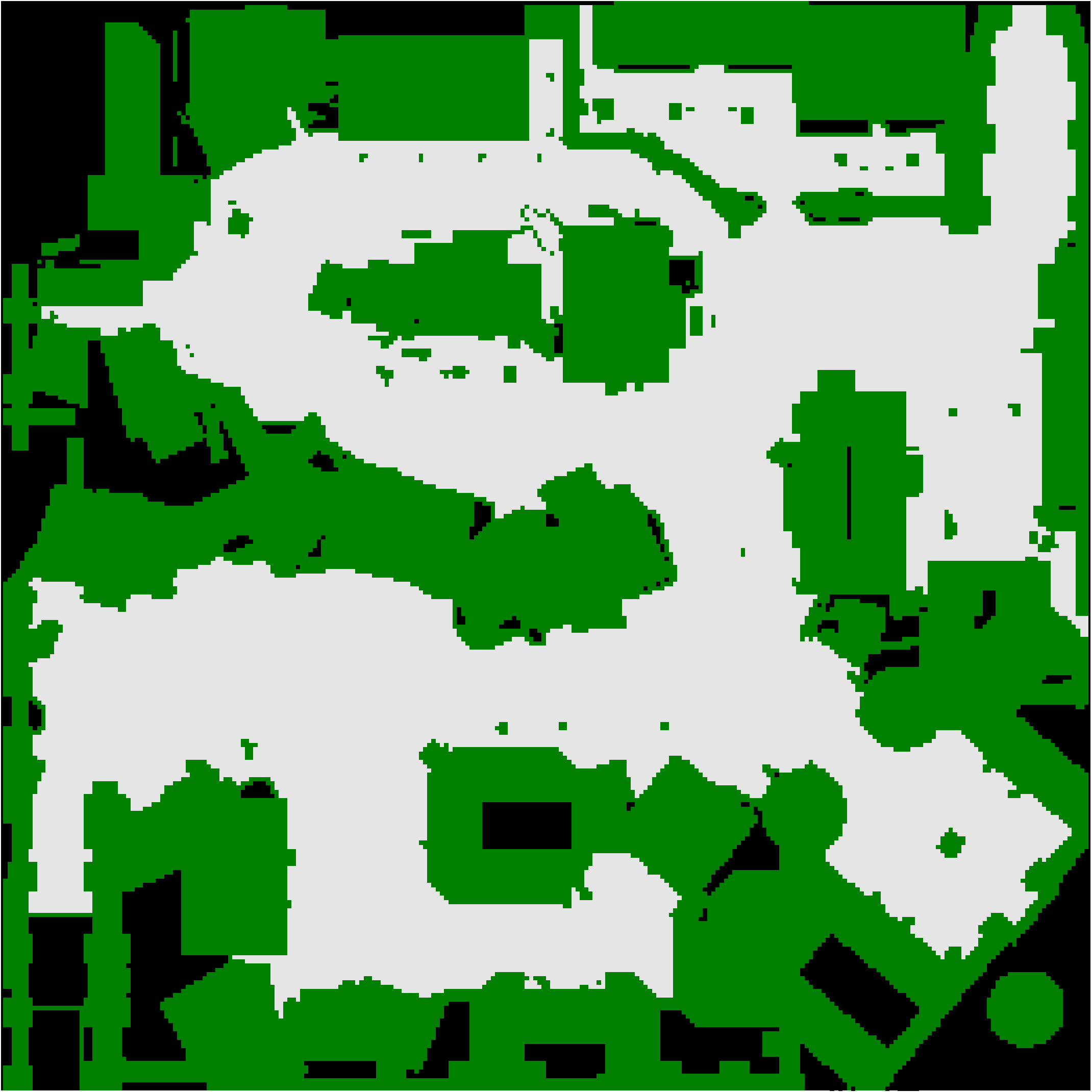}{256}{257}{28178}
        \medskip\\
      \plotStress{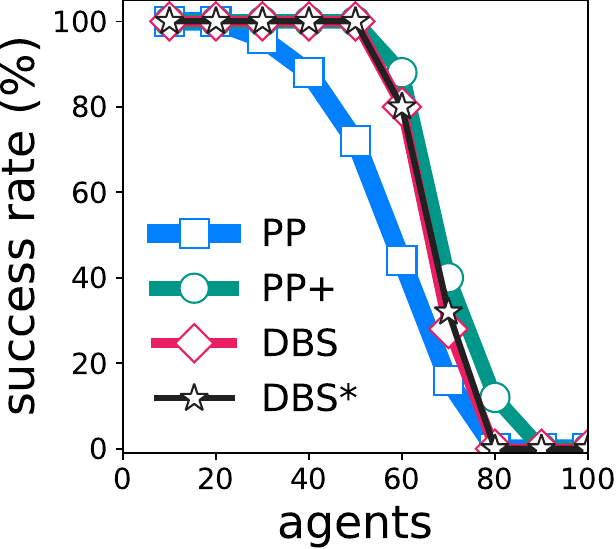}
      & \plotStress{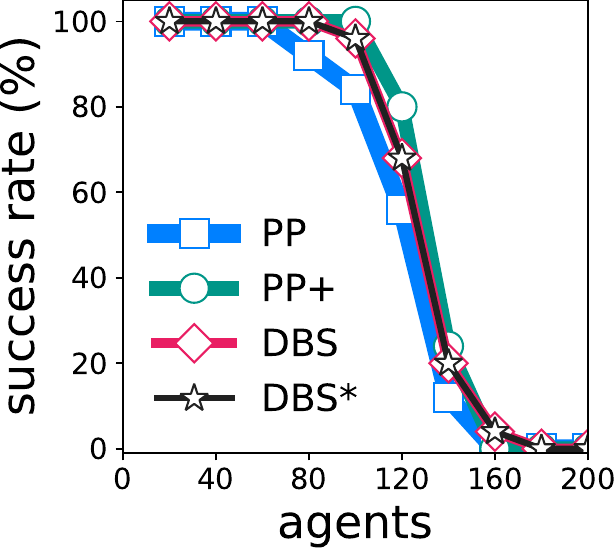}
      & \plotStress{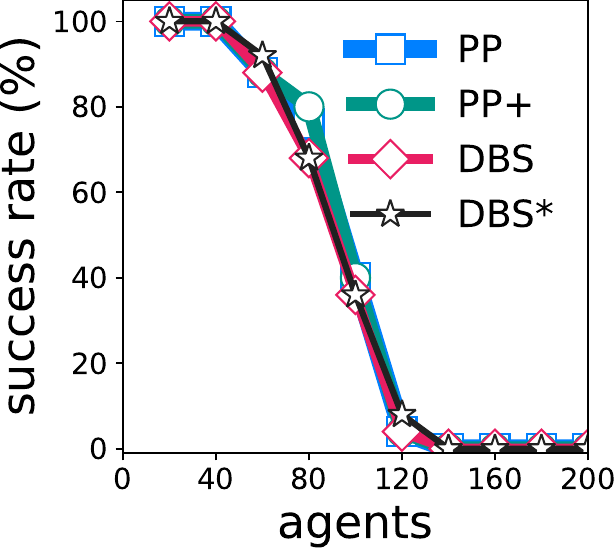}
      \\
      \plotCactus{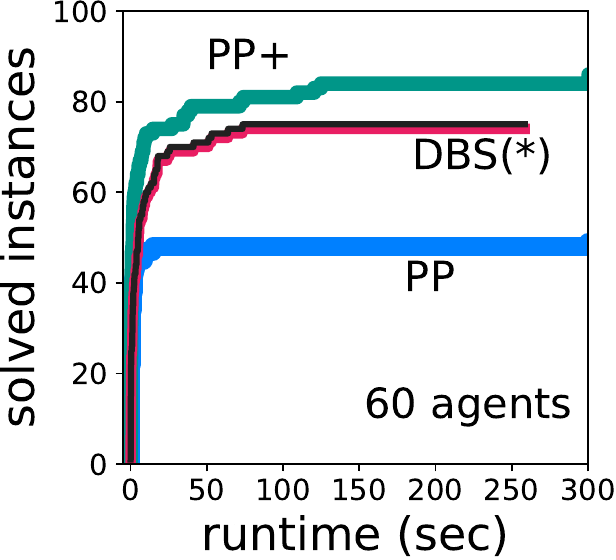}
      & \plotCactus{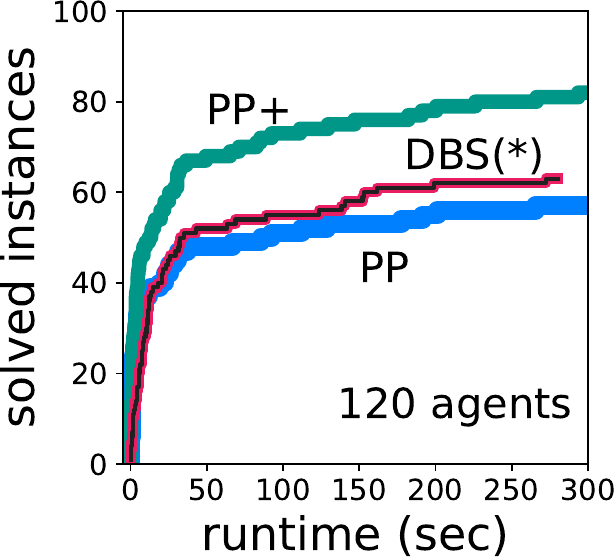}
      & \plotCactus{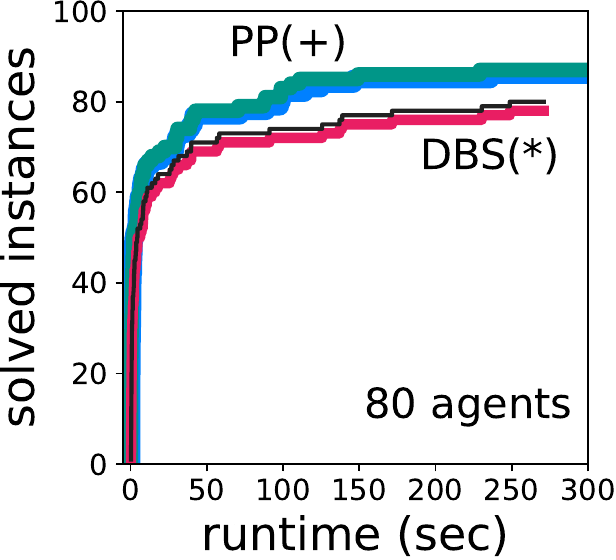}
      \\
      \plotProfiling{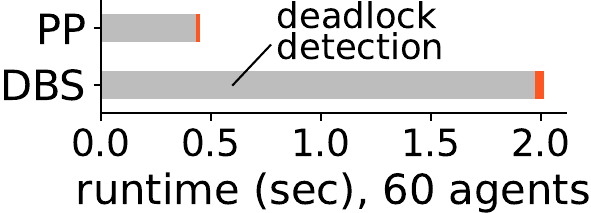}
      & \plotProfiling{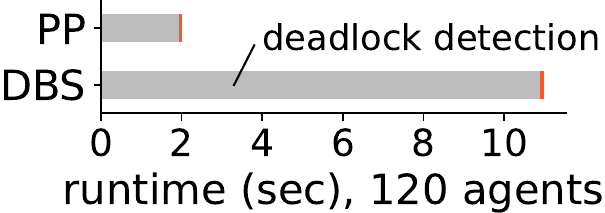}
      & \plotProfiling{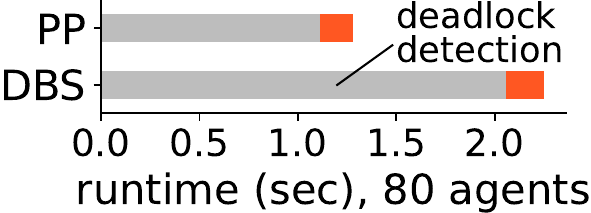}
    \end{tabular}
    \caption{
      Stress test on 4-connected grids.
      The success rate is based on 25 identical instances.
      DBS$^\ast$ includes detected instances that are unsolvable for DBS before timeout, which is not possible for PP$^{(+)}$.
      We also present accumulated runtime with a fixed number of agents over 100 instances, and runtime profiling (median) of each solver over success instances for both solvers.
    }
    \label{fig:stress-test}
  \end{figure}
}

\paragraph{Result}
Fig.~\ref{fig:stress-test} presents the results.
The main findings are:
(1)~Both solvers can solve instances with tens of agents in various maps within a reasonable time.
(2)~PP often fails due to priority orders (e.g., Fig.~\ref{fig:prioritization}) while PP$^+$ and DBS can overcome such limitations to some extent.
(3)~A bottleneck of each solver is the procedure of detecting potential cyclic deadlocks, an NP-hard problem (Lemma.~\ref{lemma:deadlock-np-comp}).
This also leads to similar success rates of PP$^+$ and DBS.

\subsection{Delay Tolerance}
\label{sec:delay-tolerance}
We next show that OTIMAPP solutions (if found) are useful in a simulated environment with stochastic delays of agents' moves built on conventional MAPF, called MAPF-DP (with Delay Probabilities)~\cite{ma2017multi}.
Given a graph and start-goal pairs for each agent, the aim of MAPF is to move agents to their goals without collisions.
Collisions occur when two agents occupy the same vertex or traverse the same edge simultaneously.
Time is discrete.
All agents synchronously take actions, i.e., either move to an adjacent vertex or stay at the current location.
MAPF-DP emulates the imperfect execution of MAPF by introducing the possibility $p_i$ of unsuccessful moves for agent $i$ (remaining there).

\paragraph{Setup}
The delay probabilities $p_i$ were chosen uniformly at random from $[0, \bar{p}]$, where $\bar{p}$ is the upper bound of $p_i$.
The higher $\bar{p}$ means that agents delay often, and vice versa.
The metric is the total traveling time of agents; smaller values mean less wasting time at runtime.
We tested the following two as baselines:
(1)~MCPs~\cite{ma2017multi} force agents to preserve order relations of visiting one vertex in an offline MAPF plan at runtime.
The plan was obtained by ECBS~\cite{barer2014suboptimal}.
(2)~Causal-PIBT~\cite{okumura2021time} is online time-independent planning, that is, each agent repeats one-step planning and action adaptively to surrounding current situations.
The other details are in the Appendix.

\paragraph{Result}
Table~\ref{table:result-delay} shows that the execution of OTIMAPP solutions outperforms the alternatives.
This is because:
(1)~Unlike MCPs, OTIMAPP solutions are free from temporal dependencies of offline plans that one agent delays are possibly fatal.
(2)~Unlike Causal-PIBT, agents follow long-term plans and avoid possible congested locations.

\paragraph{Discussion}
Although finding OTIMAPP solutions is challenging, Table~\ref{table:result-delay} motivates us to compute them.
Meanwhile, the other approaches can solve larger instances with more agents (e.g., $|A|=200$) and with much smaller planning time than solving OTIMAPP.
Moreover, there are situations where OTIMAPP has no solutions while the others can find feasible plans because OTIMAPP assumes no intervention at runtime.
One future direction pursues to fill these gaps.

{
  \setlength{\tabcolsep}{1.2mm}
  \newcommand{\ci}[1]{\tiny(#1)}
  \newcommand{\w}[1]{\textbf{#1}}
  \begin{table}[t]
    \centering
    {
    \small
    \begin{tabular}{lrlrlrl}
      \toprule
      $|A|=35$
      & \multicolumn{2}{c}{$\bar{p}=0.2$}
      & \multicolumn{2}{c}{$\bar{p}=0.5$}
      & \multicolumn{2}{c}{$\bar{p}=0.8$}
      \\\midrule
      MCPs+ECBS & 1015 & \ci{1004,1026} & 1422 & \ci{1404,1440} & 2551 & \ci{2507,2596}\\
      Causal-PIBT & 986 & \ci{976,995} & 1238 & \ci{1225,1250} & 1841 & \ci{1816,1866} \\
      OTIMAPP & \w{941} & \ci{931,951} & \w{1178} & \ci{1165,1190} & \w{1730} & \ci{1707,1752}
      \\\toprule
      $\bar{p}=0.5$
      & \multicolumn{2}{c}{$|A|=20$}
      & \multicolumn{2}{c}{$|A|=40$}
      & \multicolumn{2}{c}{$|A|=60$}
      \\\midrule
      MCPs+ECBS & 724 & \ci{711,736} & 1698 & \ci{1678,1716} & 2938 & \ci{2911,2964} \\
      Causal-PIBT & 662 & \ci{653,671} & 1466 & \ci{1453,1479} & 2425 & \ci{2405,2444} \\
      OTIMAPP & \w{639} & \ci{631,648} & \w{1395} & \ci{1383,1408} & \w{2328} & \ci{2311,2345}
      \\ \bottomrule
    \end{tabular}
    }
    \caption{
      Total traveling time on MAPF-DP.
      All settings used \field{random-32-32-10}.
      For each setting, we first picked up 10 instances that OTIMAPP solutions were found by PP$^+$.
      For each instance and approach, we then performed 50 trials while changing the random seed.
      Thus, the scores are means on 500 executions, accompanied with 95\% confidence intervals.
      \emph{upper}:~Results of changing $\bar{p}$ while fixing $|A|$.
      \emph{lower}:~Results of changing $|A|$ while fixing $\bar{p}$.
      Note that the probability that someone delays increases with more agents.
    }
    \label{table:result-delay}
  \end{table}
}

{
  \newcommand{\colwidth}{0.49\linewidth}
  \setlength{\tabcolsep}{1pt}
  \begin{figure}[t]
    \centering
    \begin{tabular}{cc}
      \begin{minipage}{\colwidth}
        \centering
        \includegraphics[width=1\hsize]{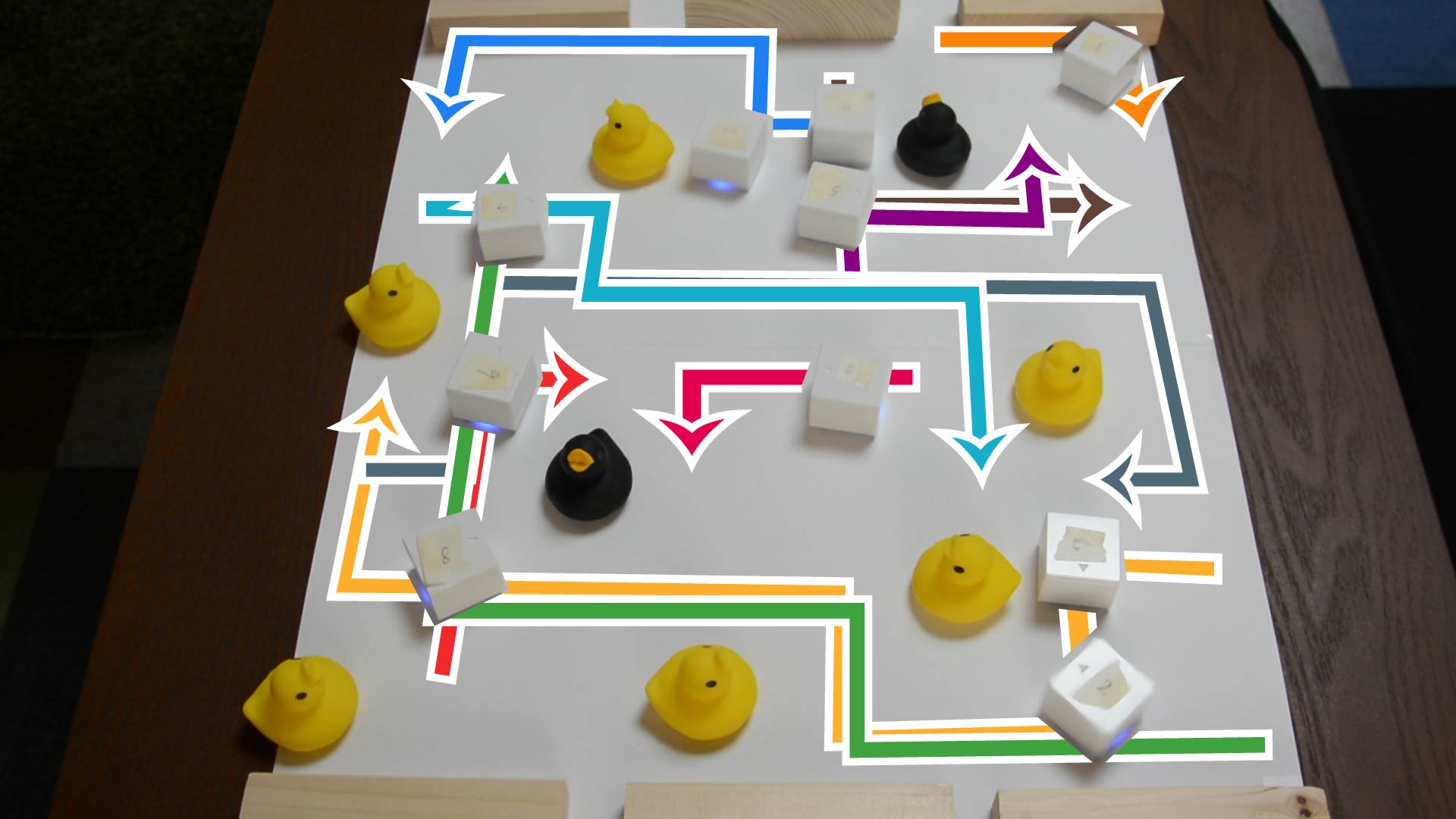}
      \end{minipage}
      &
      \begin{minipage}{\colwidth}
        \centering
        \includegraphics[width=1\hsize]{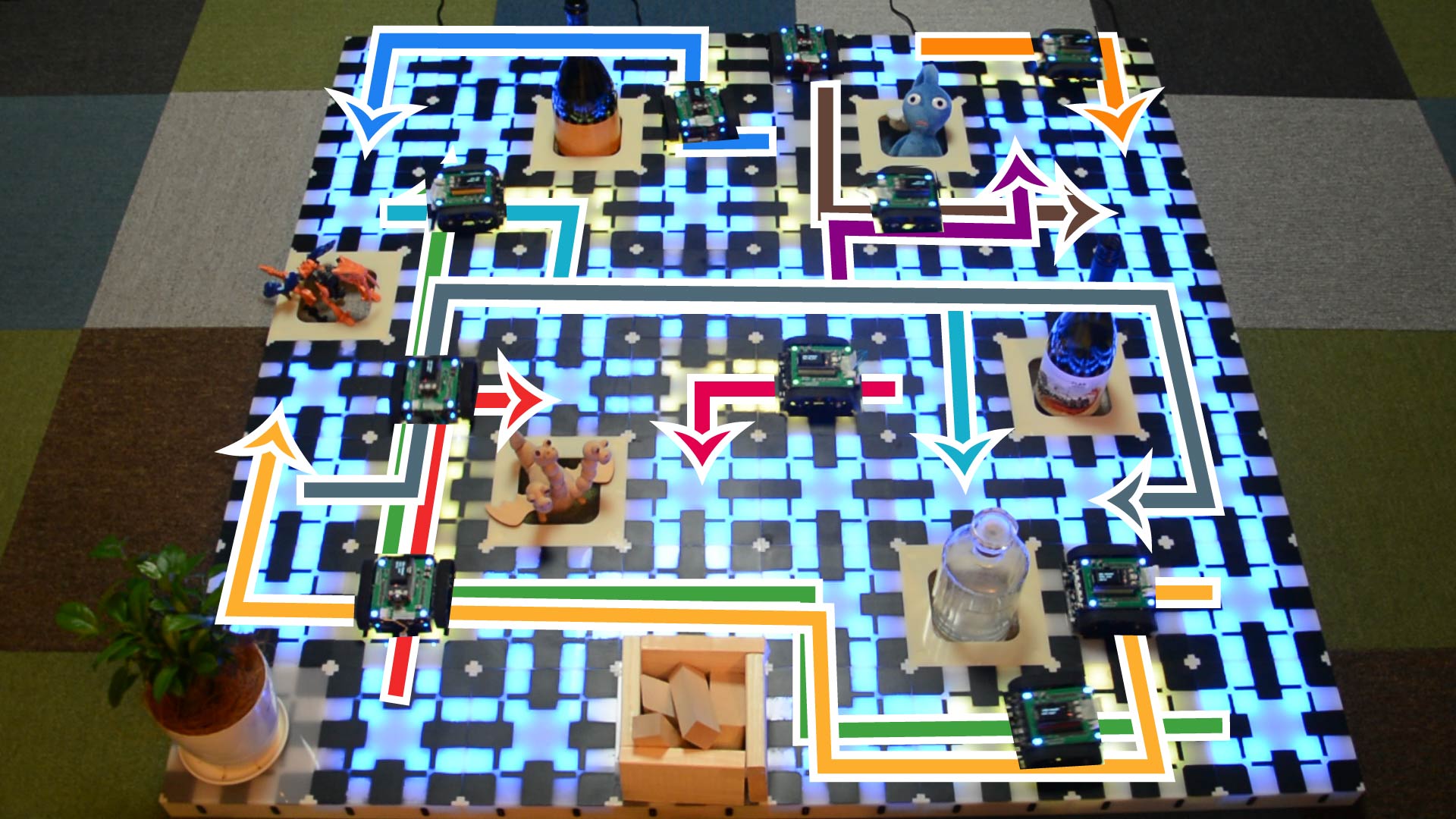}
      \end{minipage}
    \end{tabular}
    \caption{
      An OTIMAPP execution with $10$~robots in an $8\times 8$ grid.
      Colored arrows represent an OTIMAPP solution.
    }
    \label{fig:demo}
  \end{figure}
}

\subsection{Robot Demonstrations}
\label{sec:demo}
We present two OTIMAPP execution styles:
(1)~a \emph{centralized} control using the \emph{toio} robots (\url{https://toio.io}) and
(2)~a \emph{decentralized} one with only local interactions using a multi-robot platform~\cite{kameyama2021active}.
A solution was obtained by DBS.
Figure~\ref{fig:demo} is snapshots.
A video is available online.
In both cases, robots move without any synchronization procedures but are ensured to eventually reach their goals thanks to the nature of OTIMAPP.
Moreover, for the latter, any actor has no methods to know the entire configuration, which cannot be addressed by conventional execution strategies.

\section{Related Work}
A \emph{deadlock}~\cite{coffman1971system} is a widely recognized phenomenon not limited to robotics;
a system state that several components claim resources that others hold, then block each other permanently.
Strategies to cope with deadlocks are categorized into prevention, detection/recovery, and avoidance~\cite{silberschatz2006operating,fanti2004deadlock};
\emph{OTIMAPP is prevention}.
A non-deadlock state that is ``inevitable'' to reach deadlocks is called \emph{unsafe}~\cite{silberschatz2006operating}.
Meanwhile, reachable deadlocks of OTIMAPP correspond to states that are ``possible'' to reach deadlocks.
The notion of potential terminal deadlock is related to well-formed instances of MAPF~\cite{vcap2015prioritized}, that is, for each start-goal pair, a path exists that traverses no other starts and goals.
The notion of reachable cyclic deadlock is mentioned as nonlive states/sets for deadlock management in automated manufacturing systems~\cite{fanti2004deadlock} or in a multi-robot scheduling problem~\cite{mannucci2021provably}.

The \emph{multi-agent pathfinding (MAPF)} problem~\cite{stern2019def} aims at finding a set of collision-free paths on a graph.
Many studies on MAPF consider timing uncertainties because they are inevitable in multi-agent scenarios.
However, current methods largely rely on additional assumptions on the travel speed of agents or assume delays to follow some probability distributions~\cite{wagner2017path,mansouri2019multi,peltzer2020stt,atzmon2020probabilistic}.
Failing to represent the inherent uncertainty in the domain means the system behavior can be unpredictable.
Alternative approaches are online intervention during execution, e.g., forcing agents to preserve temporal dependencies of offline planning via communication~\cite{ma2017multi,honig2019persistent,atzmon2020robust}.
Another direction is online time-independent planning~\cite{okumura2021time} that incrementally moves agents based on current situations.
OTIMAPP shares the concept of time independence but aims at offline planning without or less runtime effort.

In graph theory, the (vertex) disjoint path problem and its variants~\cite{robertson1985disjoint} are partly related to ours in the sense that a set of disjoint paths clearly satisfies the solution condition of OTIMAPP, but the reverse does not.

\section{Conclusion}
This paper studied a novel path planning problem called OTIMAPP, motivated by the nature of distributed environments (i.e., timing uncertainties) that multi-agent systems must address.
We focused on robotic applications in evaluation but believe that OTIMAPP can be leveraged to other resource allocation problems with mutual exclusion, e.g, distributed databases, which is our future direction.
\section*{Acknowledgments}
We thank the anonymous reviewers for their many insightful comments and suggestions.
This work was partly supported by JSPS KAKENHI Grant Numbers~20J23011, 21K11748, and 21H03423.
Keisuke Okumura thanks the support of the Yoshida Scholarship Foundation.

\bibliographystyle{named}
\bibliography{ref}

\appendix

\setcounter{theorem}{0}

\section*{Appendix}
We complement omitted proofs~(Sec.~\ref{sec:proof:solution}, \ref{sec:proof:complexity}, and \ref{sec:proof:dbs}), the detailed procedure of detecting potential cyclic deadlocks (Sec.~\ref{sec:detection-details}), additional results of stress test on random graphs (Sec.~\ref{sec:stress-test-random}), and the details of experimental setups (Sec.~\ref{sec:exp-detail}).

\section{Proof of Solution Analysis}
\label{sec:proof:solution}

\begin{theorem*}[\ref{thrm:necessary-sufficient}; necessary and sufficient condition]
  Given an OTIMAPP instance, a set of path $\{ \path{1}, \ldots, \path{N} \}$ is a feasible solution if and only if there are:
  \begin{itemize}
    \setlength{\itemsep}{0pt}
    \item No reachable terminal deadlocks.
    \item No reachable cyclic deadlocks.
  \end{itemize}
\end{theorem*}
\begin{proof}
  Without ``no reachable terminal deadlocks,'' there is an execution that one agent arrives at its goal and remains there; disturbing the progression of another agent.
  Without  ``no reachable cyclic deadlocks,'' a cyclic deadlock might occur and those agents stop the progression.
  Hence those two are necessary.

  We now prove that the two conditions are sufficient. Given a solution candidate $\{ \path{1}, \ldots, \path{N} \}$ with no reachable deadlocks, consider the potential function $\phi \defeq \sum_{i \in A} (|\path{i}| - \clock_i)$ defined over a configuration $\{ \clock_1, \ldots, \clock_N \}$.
  Observe that $\phi$ is non-increasing and $\phi=0$ means that all agents have reached their goals.
  Furthermore, when $\phi > 0$, $\phi$ is guaranteed to decrease if each agent is activated at least once.
  We explain this as follows.

  Suppose contrary that $\phi (\neq 0)$ does not differ for the period.
  Since $\phi \neq 0$, there are agents whose progress indexes are less than the maximum values.
  Let them $B \subseteq A$.
  For an agent $i \in B$, \loc{i}{\clock_i+1} is occupied by another agent $j \in B$, according to ``no reachable terminal deadlocks,'' otherwise, $i$ moves there.
  This is the same for $j$, i.e., there is an agent $k \in B$ such that $\loc{j}{\clock_j+1} = \loc{k}{\clock_k}$.
  By induction, this sequence of agents must form a cycle somewhere, i.e., occurring a cyclic deadlock; however, this contradicts ``no reachable cyclic deadlocks.''

  Each agent is activated at least once in a sufficiently long period due to the fair assumption, deriving the statement.
\end{proof}

\section{Proofs of Computational Complexity}
\label{sec:proof:complexity}

\begin{theorem*}[\ref{thrm:np-hard-undirected}; complexity on undirected graphs]
  For OTIMAPP on \emph{undirected} graphs, it is NP-hard to find a feasible solution with \emph{simple} paths.
\end{theorem*}
\begin{proof}
  We add a new gadget, which makes an undirected edge to a virtually directed one, to the proof of the NP-hardness on digraphs (Thm.~\ref{thrm:np-hard-directed}).
  We derive the claim by replacing all directed edges, except for bidirectional edges, with this gadget.

  Figure~\ref{fig:3-sat-finding} (right) is it, including two new agents: $z_1$ and $z_2$.
  In this gadget, any agents outside of the gadget are allowed to move only the direction from $u$ to $v$.
  Assume contrary that, one agent ($\neq z_1, z_2$) takes a path through $v$ to $u$ within this gadget, and there is another path from $v$ to $u$.
  To avoid cyclic deadlocks, $z_1$ and $z_2$ must move toward the left side, exit the gadget from $u$, use another path to enter $v$, and eventually reach their goals.
  In these paths, $z_1$ can arrive at its goal earlier than that of $z_2$, contradicting ``no reachable terminal deadlocks'' in the necessary and sufficient condition (Thm.~\ref{thrm:necessary-sufficient}).
  Therefore, these paths are invalid.

  The size of the OTIMAPP instance is still polynomial on the 3-SAT formula.
  Thus, we conclude the statement.

  Note that, if we allow non-simple paths, it might be possible for other agents to move through the way from $v$ to $u$.
  This is because, even though $z_1$ and $z_2$ temporarily leave from the gadget via $u$, we can construct paths that $z_1$ always arrive at its goal after $z_2$'s arrival, as illustrated in Fig.~\ref{fig:undirected-counterexample}.
\end{proof}

\begin{figure}[t]
  \centering
  \begin{tikzpicture}
    \newcommand{\edgesize}{0.3cm}
    \newcommand{\edgesizehalf}{0.15cm}
    \newcommand{\edgesizethreetwo}{0.35cm}
    \newcommand{\longedgesize}{1.2cm}
    \tikzset{
      satnode/.style={vertex, minimum size=0.15cm, fill=black},
      start-node/.style={vertex, minimum size=0.33cm,rectangle},
      goal-node/.style={vertex, minimum size=0.33cm}
    }
    \tiny
    {
      \node[satnode, label=above:{\small $u$}](v1) at (0.0, 0.0) {};
      \node[satnode, right=\edgesize of v1](v2) {};
      \node[start-node, right=\edgesize of v2](v3) {$z_2$};
      \node[satnode, right=\edgesize of v3](v4) {};
      \node[satnode, right=\edgesize of v4,label=above:{\small $v$}](v5) {};
      \node[start-node, above=\edgesize of v2](v6) {$z_1$};
      \node[goal-node, below=\edgesize of v4](v7) {$z_1$};
      \node[goal-node, below=\edgesize of v7](v8) {$z_2$};
      \foreach \u / \v in {v1/v2,v2/v3,v3/v4,v4/v5,v6/v2,v4/v7,v7/v8}
      \draw[line] (\u) -- (\v);
      \draw[line, densely dotted] (v1) -- ($(v1) + (-0.5, 0.0)$);
      \draw[line, densely dotted] (v5) -- ($(v5) + (0.5, 0.0)$);
      \node[start-node, label=right:{start}](label-s) at (3.0, -0.4){};
      \node[goal-node, below=0.1cm of label-s, label=right:{goal}](label-g) {};
    }
    {
      \draw[line,->,orange]
      ($(v6)+(-0.1,-0.1)$) --
      ($(v2)+(-0.1, 0.1)$) --
      ($(v1)+(-0.3, 0.1)$) --
      ($(v1)+(-0.3, -1.8)$) --
      ($(v5)+(0.3, -1.8)$) --
      ($(v5)+(0.3, 0.05)$) --
      ($(v1)+(-0.25, 0.05)$) --
      ($(v1)+(-0.25, -1.75)$) --
      ($(v5)+(0.25, -1.75)$) --
      ($(v5)+(0.25, -0.05)$) --
      ($(v4)+(-0.05, -0.05)$) --
      ($(v7)+(-0.05, 0.2)$);
      \draw[line,->,cyan]
      ($(v3)+(-0.05,-0.1)$) --
      ($(v1)+(-0.2, -0.1)$) --
      ($(v1)+(-0.2, -1.7)$) --
      ($(v5)+(0.2, -1.7)$) --
      ($(v5)+(0.2, -0.1)$) --
      ($(v4)+(0.05, -0.1)$) --
      ($(v8)+(0.05, 0.2)$);
    }
    {
      \node[start-node, above=\edgesize of v2](v6) {$z_1$};
      \node[start-node, right=\edgesize of v2](v3) {$z_2$};
      \node[goal-node, below=\edgesize of v4](v7) {$z_1$};
    }
  \end{tikzpicture}
  \caption{
    Counterexample of \textit{oneway constrainer} without assmptions of simple paths.
    $z_2$ always arrives at its goal before the arrival of $z_1$ if those two agents follow colored lines.
  }
  \label{fig:undirected-counterexample}
\end{figure}
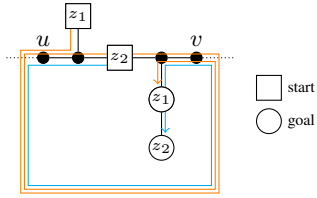
\begin{figure*}[t]
  \centering
  \begin{tikzpicture}
    \newcommand{\edgesize}{1.6cm}
    \newcommand{\hedgex}{0.75}
    \newcommand{\hedgey}{0.6}
    \tikzset{
      sat-node/.style={vertex, minimum size=0.15cm, fill=black, anchor=center},
      start-node/.style={vertex, minimum size=0.45cm, text width=0.45cm, rectangle},
      goal-node/.style={vertex, minimum size=0.5cm},
      line-c1/.style={orange,->},
      line-c2/.style={blue,->},
      line-c3/.style={cyan,->},
      bold-line/.style={very thick},
    }
    \tiny
    \node[sat-node] (v1) at (0, 0) {};
    \node[sat-node, right=\edgesize of v1.center](v2) {};
    \node[sat-node, right=\edgesize of v2.center](v3) {};
    \node[sat-node, right=\edgesize of v3.center](v4) {};
    \node[sat-node, right=\edgesize of v4.center](v5) {};
    \node[sat-node, right=\edgesize of v5.center](v6) {};
    \node[start-node, right=\edgesize of v6.center](v7) {$z$};
    \node[sat-node](v8) at ($(v1)+(\hedgex, \hedgey)$) {};
    \node[sat-node](v9) at ($(v2)+(\hedgex, \hedgey)$) {};
    \node[sat-node](v10) at ($(v3)+(\hedgex, -\hedgey)$) {};
    \node[start-node](v11) at ($(v2.center) + (0,  1.6)$) {$c^1_2$};
    \node[start-node](v12) at ($(v2.center) + (0, -1.6)$) {$c^3_2$};
    \node[start-node, above=0.77cm of v9.center](v17) {$c^2_2$};
    \node[sat-node](v13) at ($(v9) + ( 0.2,0.4)$) {};
    \node[sat-node](v14) at ($(v3) + (-0.1,-1.1)$) {};
    \node[sat-node](v19) at ($(v3) + (-0.1, 1.2)$){};
    \node[goal-node, above right=1.3cm of v4.center](v15) {$c^1_2$};
    \node[goal-node, above right=1.3cm of v5.center](v18) {$c^2_2$};
    \node[goal-node, above right=1.3cm of v6.center](v16) {$c^3_2$};
    \node[goal-node, left=0cm of v18](v20) {$c^2_1$};
    \node[goal-node, right=0cm of v18](v21) {$c^2_3$};
    \node[goal-node](v22) at ($(v1) + (-0.2, -0.8)$) {$z$};
    \draw[line,->] (v1) -- (v22);
    \coordinate[ left=0.5cm of v1, label=left:{\small$\clubsuit$}](circle-c1);
    \coordinate[right=0.5cm of v7,label=right:{\small$\clubsuit$}](circle-c2);
    %
    {
      \draw[line,->,bold-line] (v1) to[out=60,in=180] (v8);
      \draw[line,->,bold-line] (v8) to[out=0,in=120] (v2);
      \draw[line-c1,bold-line] (v2) to[out=60,in=180] (v9);
      \draw[line-c3,bold-line] (v9) to[out=0,in=120] (v3);
      \draw[line,->,bold-line] (v3) to[out=-60,in=180] (v10);
      \draw[line,->,bold-line] (v10) to[out=0,in=-120] (v4);
      \draw[line,->,bold-line] (v1) to[out=-60,in=-120] (v2);
      \draw[line-c2,bold-line] (v2) to[out=-60,in=-120] (v3);
      \draw[line,->,bold-line] (v3) to[out=60,in=120] (v4);
    }
    {
      \draw[line,->] (v4) to[out=60,in=120] (v5);
      \draw[line-c1] (v4) -- (v5);
      \draw[line,->] (v4) to[out=-60,in=-120] (v5);
      \draw[line,->] (v5) to[out=60,in=120] (v6);
      \draw[line-c3] (v5) -- (v6);
      \draw[line,->] (v5) to[out=-60,in=-120] (v6);
      \draw[line,->] (v6) to[out=60,in=120] (v7);
      \draw[line-c2] (v6) -- (v7);
      \draw[line,->] (v6) to[out=-60,in=-120] (v7);
    }
    {
      \draw[line-c1] (v11) -- (v2);
      \draw[line-c1] (v5)  to[out=100,in=-50] (v15);
      \draw[line-c2] (v12) -- (v2);
      \draw[line-c2] (v7.north) -- (v16);
      \draw[line-c3] (v17) -- (v9);
      \draw[line-c3] (v6) to[out=100,in=-50] (v18);
      \draw[line,->] (v6) to[out=110,in=-40] (v20);
      \draw[line,->] (v6)  -- (v21);
    }
    {
      \draw[line-c1] (v9) -- (v13);
      \draw[line-c1] (v13) to[out=0,in=110] (v4);
      \draw[line-c2] (v3) -- (v14);
      \draw[line-c2] (v14) to[out=-10,in=-110] (v6);
      \draw[line-c3] (v3) -- (v19);
      \draw[line-c3] (v19) to[out=5,in=120] (v5);
    }
    {
      \draw[line,->] (circle-c1.east) -- (v1);
      \draw[line,->] (v7) -- (circle-c2.west);
   }
    {
     \node[rectangle,anchor=north west,
     minimum height=3.8cm,text width=2.0cm, draw,densely dotted,
     label=below:{\textit{variable decider}: $x_2$}] at (1.35, 1.9) {};
     \node[rectangle,anchor=north west,
     minimum height=2.3cm,text width=1.95cm, draw,densely dotted,
     label=above:{\textit{clause constrainer}:~$C^2$}] at (6.55, 1.5) {};
     \node[above right=0.5cm of v5,anchor=south] {$c^2_1$};
     \node[right=0.7cm of v5,anchor=south] {$c^2_2$};
     \node[below right=0.45cm of v5,anchor=north] {$c^2_3$};
     \node[anchor=west] at (1.9,  0.3) {true (upper)};
     \node[anchor=west] at (1.9, -0.2) {false (lower)};
     \node[start-node, label=right:{start}](label-s) at (10.5, -1.1){};
     \node[goal-node, below=0.1cm of label-s, label=right:{goal}](label-g) {};
     \node[](vacation) at (4.5, 2.0) {\textit{vacation vertex}};
     \draw[densely dashed,->] (vacation.west) -- ($(v19)+(0.1,0.1)$);
     \draw[densely dashed,->] (vacation.west) -- ($(v13)+(0.1,0.1)$);
   }
  \end{tikzpicture}
  \caption{OTIMAPP instance and solution reduced from the 3-SAT instance $(x_1 \lor x_2 \lor \lnot x_3) \land (\lnot x_1 \lor x_2 \lor x_3) \land (x_1 \lor \lnot x_2 \land \lnot x3)$.
    For visualization, we break a large circle; regard two $\clubsuit$ marks as connected.
    Omitted vertices and edges are complemented in Fig.~\ref{fig:3-sat-deadlocks-detail}.
  }
  \label{fig:3-sat-deadlocks}
\end{figure*}
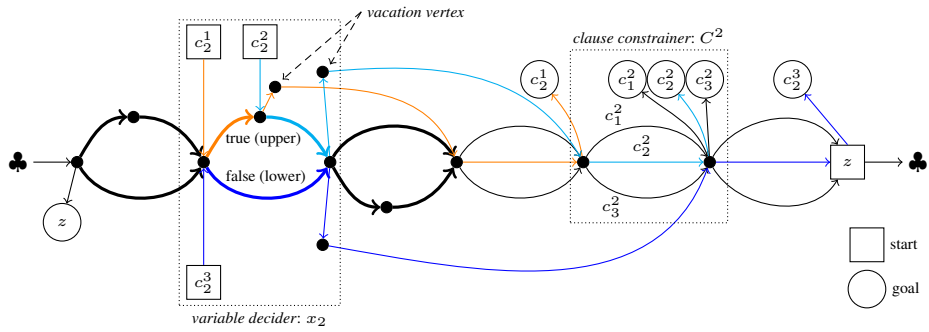
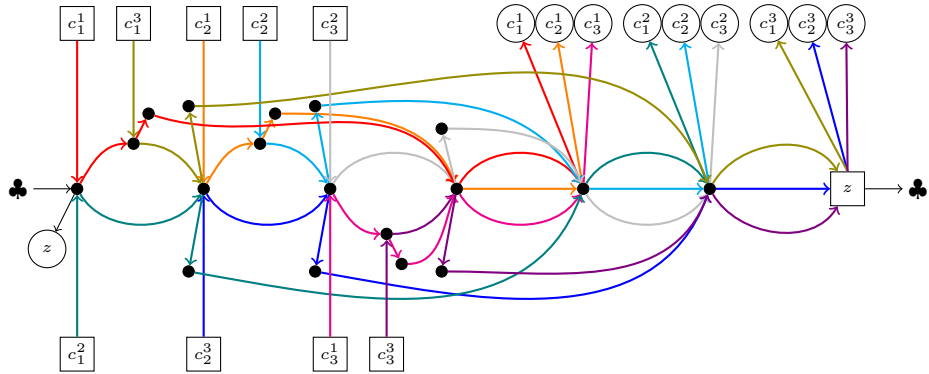
\begin{figure*}[t]
  \centering
  \begin{tikzpicture}
    \newcommand{\edgesize}{1.6cm}
    \newcommand{\hedgex}{0.75}
    \newcommand{\hedgey}{0.6}
    \tikzset{
      sat-node/.style={vertex, minimum size=0.15cm, fill=black, anchor=center},
      start-node/.style={vertex, minimum size=0.45cm, text width=0.45cm, rectangle},
      goal-node/.style={vertex, minimum size=0.5cm},
      line-c1/.style={orange,->,thick},
      line-c2/.style={blue,->,thick},
      line-c3/.style={cyan,->,thick},
      line-c4/.style={red,->,thick},
      line-c5/.style={teal,->,thick},
      line-c6/.style={olive,->,thick},
      line-c7/.style={magenta,->,thick},
      line-c8/.style={lightgray,->,thick},
      line-c9/.style={violet,->,thick},
    }
    \tiny
    \node[sat-node] (v1) at (0, 0) {};
    \node[sat-node, right=\edgesize of v1.center](v2) {};
    \node[sat-node, right=\edgesize of v2.center](v3) {};
    \node[sat-node, right=\edgesize of v3.center](v4) {};
    \node[sat-node, right=\edgesize of v4.center](v5) {};
    \node[sat-node, right=\edgesize of v5.center](v6) {};
    \node[start-node, right=\edgesize of v6.center](v7) {$z$};
    \node[sat-node](v8) at ($(v1)+(\hedgex, \hedgey)$) {};
    \node[sat-node](v9) at ($(v2)+(\hedgex, \hedgey)$) {};
    \node[sat-node](v10) at ($(v3)+(\hedgex, -\hedgey)$) {};
    \node[start-node](v11) at ($(v2.center) + (0,  2.2)$) {$c^1_2$};
    \node[start-node](v12) at ($(v2.center) + (0, -2.2)$) {$c^3_2$};
    \node[start-node](v17) at ($(v9.center) + (0, 1.6)$){$c^2_2$};
    \node[sat-node](v13) at ($(v9) + ( 0.2,0.4)$) {};
    \node[sat-node](v14) at ($(v3) + (-0.2,-1.1)$) {};
    \node[sat-node](v19) at ($(v3) + (-0.2, 1.1)$){};
    \node[sat-node](v33) at ($(v2) + (-0.2, 1.1)$){};
    \node[sat-node](v34) at ($(v8) + ( 0.2, 0.4)$){};
    \node[sat-node](v35) at ($(v2) + (-0.2, -1.1)$){};
    \node[sat-node](v36) at ($(v4) + (-0.2, -1.1)$){};
    \node[sat-node](v37) at ($(v4) + (-0.2,  0.8)$){};
    \node[sat-node](v38) at ($(v10) + ( 0.2, -0.4)$){};
    \node[goal-node](v15) at ($(v4.center) + (1.3, 2.2)$) {$c^1_2$};
    \node[goal-node, left=0cm of v15](v31) {$c^1_1$};
    \node[goal-node, right=0cm of v15](v32) {$c^1_3$};
    \node[goal-node](v16) at ($(v6.center) + (1.3, 2.2)$) {$c^3_2$};
    \node[goal-node, left=0cm of v16](v29) {$c^3_1$};
    \node[goal-node, right=0cm of v16](v30) {$c^3_3$};
    \node[goal-node](v18) at ($(v5.center) + (1.3, 2.2)$) {$c^2_2$};
    \node[goal-node, left=0cm of v18](v20) {$c^2_1$};
    \node[goal-node, right=0cm of v18](v21) {$c^2_3$};
    \node[goal-node](v22) at ($(v1) + (-0.4, -0.8)$) {$z$};
    \draw[line,->] (v1) -- (v22);
    \coordinate[ left=0.5cm of v1, label=left:{\small$\clubsuit$}](circle-c1);
    \coordinate[right=0.5cm of v7,label=right:{\small$\clubsuit$}](circle-c2);
    %
    \node[start-node](v23) at ($(v1.center) + (0, 2.2)$) {$c^1_1$};
    \node[start-node](v24) at ($(v1.center) + (0, -2.2)$) {$c^2_1$};
    \node[start-node](v25) at ($(v8.center) + (0, 1.6)$) {$c^3_1$};
    \node[start-node](v26) at ($(v3.center) + (0, -2.2)$) {$c^1_3$};
    \node[start-node](v27) at ($(v3.center) + (0, 2.2)$) {$c^2_3$};
    \node[start-node](v28) at ($(v10.center) + (0, -1.6)$) {$c^3_3$};
    %

    {
      \draw[line-c4] (v1) to[out=60,in=180] (v8);
      \draw[line-c6] (v8) to[out=0,in=120] (v2);
      \draw[line-c1] (v2) to[out=60,in=180] (v9);
      \draw[line-c3] (v9) to[out=0,in=120] (v3);
      \draw[line-c7] (v3) to[out=-60,in=180] (v10);
      \draw[line-c9] (v10) to[out=0,in=-120] (v4);
      \draw[line-c5] (v1) to[out=-60,in=-120] (v2);
      \draw[line-c2] (v2) to[out=-60,in=-120] (v3);
      \draw[line-c8] (v3) to[out=60,in=120] (v4);
    }
    {
      \draw[line-c4] (v4) to[out=60,in=120] (v5);
      \draw[line-c1] (v4) -- (v5);
      \draw[line-c7] (v4) to[out=-60,in=-120] (v5);
      \draw[line-c5] (v5) to[out=60,in=120] (v6);
      \draw[line-c3] (v5) -- (v6);
      \draw[line-c8] (v5) to[out=-60,in=-120] (v6);
      \draw[line-c6] (v6) to[out=60,in=120] (v7);
      \draw[line-c2] (v6) -- (v7);
      \draw[line-c9] (v6) to[out=-60,in=-120] (v7);
    }
    {
      \draw[line-c1] (v11) -- (v2);
      \draw[line-c1] (v5) -- (v15);
      \draw[line-c2] (v12) -- (v2);
      \draw[line-c2] (v7.north) -- (v16);
      \draw[line-c3] (v17) -- (v9);
      \draw[line-c3] (v6) -- (v18);
      \draw[line-c5] (v6) -- (v20);
      \draw[line-c8] (v6) -- (v21);
      \draw[line-c4] (v5) -- (v31);
      \draw[line-c7] (v5) -- (v32);
      \draw[line-c6] (v7.north) -- (v29);
      \draw[line-c9] (v7.north) -- (v30);
    }
    {
      \draw[line-c1] (v9) -- (v13);
      \draw[line-c1] (v13) to[out=0,in=110] (v4);
      \draw[line-c2] (v3) -- (v14);
      \draw[line-c2] (v14) to[out=-10,in=-110] (v6);
      \draw[line-c3] (v3) -- (v19);
      \draw[line-c3] (v19) to[out=5,in=120] (v5);
      \draw[line-c4] (v34) to[out=-18,in=115] (v4);
      \draw[line-c5] (v35) to[out=-10,in=-110] (v5);
      \draw[line-c6] (v33) to[out=0,in=110] (v6);
      \draw[line-c7] (v38) to[out=0,in=-110] (v4);
      \draw[line-c8] (v37) to[out=0,in=110] (v5);
      \draw[line-c9] (v36) to[out=0,in=-110] (v6);
    }
    {
      \draw[line,->] (circle-c1.east) -- (v1);
      \draw[line,->] (v7) -- (circle-c2.west);
    }
    {
      \draw[line-c4] (v23) -- (v1);
      \draw[line-c5] (v24) -- (v1);
      \draw[line-c6] (v25) -- (v8);
      \draw[line-c7] (v26) -- (v3);
      \draw[line-c8] (v27) -- (v3);
      \draw[line-c9] (v28) -- (v10);
      \draw[line-c6] (v2) -- (v33);
      \draw[line-c4] (v8) -- (v34);
      \draw[line-c5] (v2) -- (v35);
      \draw[line-c9] (v4) -- (v36);
      \draw[line-c8] (v4) -- (v37);
      \draw[line-c7] (v10) -- (v38);
    }

  \end{tikzpicture}
  \caption{Complete version of Fig.~\ref{fig:3-sat-deadlocks} for $(x_1 \lor x_2 \lor \lnot x_3) \land (\lnot x_1 \lor x_2 \lor x_3) \land (x_1 \lor \lnot x_2 \land \lnot x3)$.
    Each color corresponds to a path for each agent.
  }
  \label{fig:3-sat-deadlocks-detail}
\end{figure*}

\begin{lemma*}[\ref{lemma:deadlock-np-comp}; complexity of detecting cyclic deadlocks]
  Determining whether a set of paths contains either reachable or potential cyclic deadlocks is NP-complete.
\end{lemma*}
\begin{proof}
  The proof is a reduction from the 3-SAT problem, i.e., constructing a combination of an OTIMAPP instance and a set of paths such that potential cyclic deadlocks exist if and only if the corresponding formula is satisfiable.
  We show the case of directed graphs.
  The proof procedure applies to the undirected case without modifications.
  In addition, all potential cyclic deadlocks are reachable in the translated problem.
  The reduction is done in polynomial time, deriving the NP-hardness of detecting both reachable and potential cyclic deadlocks.
  Since a potential cyclic deadlock can be verified in polynomial time, and since a reachable cyclic deadlock can be verified in polynomial time with an execution schedule, they are NP-complete.

  We now explain how to translate the 3-SAT formula to the OTIMAPP instance and the corresponding set of paths.
  Without loss of generality, we assume that all variables appear positively and negatively in the formula.
  Throughout the proof, we use the following example.
  \begin{align*}
    (x_1 \lor x_2 \lor \lnot x_3) \land (\lnot x_1 \lor x_2 \lor x_3) \land (x_1 \lor \lnot x_2 \land \lnot x_3)
  \end{align*}
  Its outcome is partially depicted in Fig.~\ref{fig:3-sat-deadlocks}.
  The complete version is presented in Fig.~\ref{fig:3-sat-deadlocks-detail}.

\medskip
\noindent
\emph{A. Construction of an OTIMAPP instance and a set of paths}
For each literal in each clause, one \emph{literal agent} is introduced.
We denote by $c^j_k$ a literal agent for the $k$-th literal in $j$-th clause $C^j$ in the formula.
We also use one special agent $z$.

Next, consider two gadgets: \emph{variable decider} and \emph{clause constrainer}.
Note that they are different from those used in the proof of Thm.~\ref{thrm:np-hard-directed};
however, their intuitions are similar and we use the same names.

The variable decider determines whether a variable $x_i$ occurs positively or negatively.
For each variable one gadget is introduced.
All literal agents for $x_i$ (i.e., either $x_i$ or $\lnot x_i$) start from vertices in this gadget.
The gadget contains two paths: an \emph{upper} path, corresponding to assign true to $x_i$, and a \emph{lower} path, corresponding to assign false to $x_i$.
Positive literals are connected to the upper path.
Negative literals are connected to the lower path.
For instance, $x_2$ has three literal agents: $c^1_2$ ($x_2$), $c^2_2$ ($x_2$), and $c^3_2$ ($\lnot~x_2$).
In Fig.~\ref{fig:3-sat-deadlocks}, we highlight the upper and the lower paths by bold lines.
$c^1_2$ and $c^2_2$ are connected to the upper path while $c^3_2$ is connected to the lower path.
Each literal agent uses one edge in the upper/lower path and moves to a clause constrainer via one \emph{vacation} vertex.

The clause constrainer contains all goals of the literal agents in the clause.
Three edges are used to reach the goals.
Each edge is for each literal agent.
For instance, the clause constrainer of $C^2$ contains the goals of $c^2_1$, $c^2_2$, and $c^2_3$.
In Fig.~\ref{fig:3-sat-deadlocks}, three edges are annotated with the agent's name.
$c_2^2$ is supposed to use the colored middle one.
Note that we use multiple edges for simplicity.
It is not hard to convert the gadget to a simple graph version, as shown immediately later of this proof.

As a result, all literal agents take six edges to reach their goals.
This is visualized by colored edges in Fig.~\ref{fig:3-sat-deadlocks} and Fig.~\ref{fig:3-sat-deadlocks-detail}.
The special agent~$z$ uses two edges to reach its goal, through $\clubsuit$ marks in the figure.
We finish the description of how to construct the OTIMAPP instance and the corresponding set of paths.
The remaining part is to show these paths contain potential/reachable cyclic deadlocks if and only if the formula is satisfiable.
This translation from the formula is clearly done in polynomial time.

\newcommand{\true}{\m{\mathit{true}}}
\newcommand{\false}{\m{\mathit{false}}}

\medskip
\noindent
\emph{B. A potential cyclic deadlock exists if the formula is satisfiable.}
To see this, observe that if a potential cyclic deadlock exists, agents must try to use;
(a)~either an upper or a lower path for each variable decider,
(b)~one edge for each clause constrainer, and
(c)~the edge for $z$ ($\clubsuit$).

{
\newcommand{\edgesize}{1.6cm}
\newcommand{\hedgex}{0.75}
\newcommand{\hedgey}{0.6}
\newcommand{\body}{
  \tikzset{
    sat-node/.style={vertex, minimum size=0.15cm, fill=black, anchor=center},
    start-node/.style={vertex, minimum size=0.45cm, text width=0.45cm, rectangle},
    goal-node/.style={vertex, minimum size=0.5cm},
    line-c1/.style={orange,->},
    line-c2/.style={blue,->},
    line-c3/.style={cyan,->},
    line-c4/.style={red,->},
    line-c5/.style={teal,->},
    line-c6/.style={olive,->},
    line-c7/.style={magenta,->},
    line-c8/.style={lightgray,->},
    line-c9/.style={violet,->},
    line-c10/.style={pink,->},
  }
  \tikzset{
    agent/.style={start-node},
    left-agent/.style={start-node,fill=black,fill opacity=0.6},
  }
    \tiny
    \node[sat-node] (v1) at (0, 0) {};
    \node[sat-node, right=\edgesize of v1.center](v2) {};
    \node[sat-node, right=\edgesize of v2.center](v3) {};
    \node[sat-node, right=\edgesize of v3.center](v4) {};
    \node[sat-node, right=\edgesize of v4.center](v5) {};
    \node[sat-node, right=\edgesize of v5.center](v6) {};
    \node[agent,line-c10,text=white,right=\edgesize of v6.center](v7) {$z$};
    \node[sat-node](v8) at ($(v1)+(\hedgex, \hedgey)$) {};
    \node[sat-node](v9) at ($(v2)+(\hedgex, \hedgey)$) {};
    \node[sat-node](v10) at ($(v3)+(\hedgex, -\hedgey)$) {};
    \node[start-node](v11) at ($(v2.center) + (0,  2.2)$) {$c^1_2$};
    \node[start-node](v12) at ($(v2.center) + (0, -2.2)$) {$c^3_2$};
    \node[start-node](v17) at ($(v9.center) + (0, 1.6)$){$c^2_2$};
    \node[sat-node](v13) at ($(v9) + ( 0.2,0.4)$) {};
    \node[sat-node](v14) at ($(v3) + (-0.2,-1.1)$) {};
    \node[sat-node](v19) at ($(v3) + (-0.2, 1.1)$){};
    \node[sat-node](v33) at ($(v2) + (-0.2, 1.1)$){};
    \node[sat-node](v34) at ($(v8) + ( 0.2, 0.4)$){};
    \node[sat-node](v35) at ($(v2) + (-0.2, -1.1)$){};
    \node[sat-node](v36) at ($(v4) + (-0.2, -1.1)$){};
    \node[sat-node](v37) at ($(v4) + (-0.2,  0.8)$){};
    \node[sat-node](v38) at ($(v10) + ( 0.2, -0.4)$){};
    \node[goal-node](v15) at ($(v4.center) + (1.3, 2.2)$) {$c^1_2$};
    \node[goal-node, left=0cm of v15](v31) {$c^1_1$};
    \node[goal-node, right=0cm of v15](v32) {$c^1_3$};
    \node[goal-node](v16) at ($(v6.center) + (1.3, 2.2)$) {$c^3_2$};
    \node[goal-node, left=0cm of v16](v29) {$c^3_1$};
    \node[goal-node, right=0cm of v16](v30) {$c^3_3$};
    \node[goal-node](v18) at ($(v5.center) + (1.3, 2.2)$) {$c^2_2$};
    \node[goal-node, left=0cm of v18](v20) {$c^2_1$};
    \node[goal-node, right=0cm of v18](v21) {$c^2_3$};
    \node[goal-node](v22) at ($(v1) + (-0.4, -0.8)$) {$z$};
    \draw[line,->] (v1) -- (v22);
    \coordinate[ left=0.5cm of v1, label=left:{\small$\clubsuit$}](circle-c1);
    \coordinate[right=0.5cm of v7,label=right:{\small$\clubsuit$}](circle-c2);
    %
    \node[start-node](v23) at ($(v1.center) + (0, 2.2)$) {$c^1_1$};
    \node[start-node](v24) at ($(v1.center) + (0, -2.2)$) {$c^2_1$};
    \node[start-node](v25) at ($(v8.center) + (0, 1.6)$) {$c^3_1$};
    \node[start-node](v26) at ($(v3.center) + (0, -2.2)$) {$c^1_3$};
    \node[start-node](v27) at ($(v3.center) + (0, 2.2)$) {$c^2_3$};
    \node[start-node](v28) at ($(v10.center) + (0, -1.6)$) {$c^3_3$};
    %

    {
      \draw[line-c4] (v1) to[out=60,in=180] (v8);
      \draw[line-c6] (v8) to[out=0,in=120] (v2);
      \draw[line-c1] (v2) to[out=60,in=180] (v9);
      \draw[line-c3] (v9) to[out=0,in=120] (v3);
      \draw[line-c7] (v3) to[out=-60,in=180] (v10);
      \draw[line-c9] (v10) to[out=0,in=-120] (v4);
      \draw[line-c5] (v1) to[out=-60,in=-120] (v2);
      \draw[line-c2] (v2) to[out=-60,in=-120] (v3);
      \draw[line-c8] (v3) to[out=60,in=120] (v4);
    }
    {
      \draw[line-c4] (v4) to[out=60,in=120] (v5);
      \draw[line-c1] (v4) -- (v5);
      \draw[line-c7] (v4) to[out=-60,in=-120] (v5);
      \draw[line-c5] (v5) to[out=60,in=120] (v6);
      \draw[line-c3] (v5) -- (v6);
      \draw[line-c8] (v5) to[out=-60,in=-120] (v6);
      \draw[line-c6] (v6) to[out=60,in=120] (v7);
      \draw[line-c2] (v6) -- (v7);
      \draw[line-c9] (v6) to[out=-60,in=-120] (v7);
    }
    {
      \draw[line-c1] (v11) -- (v2);
      \draw[line-c1] (v5) -- (v15);
      \draw[line-c2] (v12) -- (v2);
      \draw[line-c2] (v7.north) -- (v16);
      \draw[line-c3] (v17) -- (v9);
      \draw[line-c3] (v6) -- (v18);
      \draw[line-c5] (v6) -- (v20);
      \draw[line-c8] (v6) -- (v21);
      \draw[line-c4] (v5) -- (v31);
      \draw[line-c7] (v5) -- (v32);
      \draw[line-c6] (v7.north) -- (v29);
      \draw[line-c9] (v7.north) -- (v30);
    }
    {
      \draw[line-c1] (v9) -- (v13);
      \draw[line-c1] (v13) to[out=0,in=110] (v4);
      \draw[line-c2] (v3) -- (v14);
      \draw[line-c2] (v14) to[out=-10,in=-110] (v6);
      \draw[line-c3] (v3) -- (v19);
      \draw[line-c3] (v19) to[out=5,in=120] (v5);
      \draw[line-c4] (v34) to[out=-18,in=115] (v4);
      \draw[line-c5] (v35) to[out=-10,in=-110] (v5);
      \draw[line-c6] (v33) to[out=0,in=110] (v6);
      \draw[line-c7] (v38) to[out=0,in=-110] (v4);
      \draw[line-c8] (v37) to[out=0,in=110] (v5);
      \draw[line-c9] (v36) to[out=0,in=-110] (v6);
    }
    {
      \draw[line-c10] (circle-c1.east) -- (v1);
      \draw[line-c10] (v7) -- (circle-c2.west);
    }
    {
      \draw[line-c4] (v23) -- (v1);
      \draw[line-c5] (v24) -- (v1);
      \draw[line-c6] (v25) -- (v8);
      \draw[line-c7] (v26) -- (v3);
      \draw[line-c8] (v27) -- (v3);
      \draw[line-c9] (v28) -- (v10);
      \draw[line-c6] (v2) -- (v33);
      \draw[line-c4] (v8) -- (v34);
      \draw[line-c5] (v2) -- (v35);
      \draw[line-c9] (v4) -- (v36);
      \draw[line-c8] (v4) -- (v37);
      \draw[line-c7] (v10) -- (v38);
    }
}
\newcommand{\stepone}{
  \node[agent,line-c4,text=white] at (v34) {$c^1_1$};
  \node[left-agent] at (v23) {$c^1_1$};
  \node[agent,line-c1,text=white] at (v13) {$c^1_2$};
  \node[left-agent] at (v11) {$c^1_2$};
  \node[agent,line-c3,text=white] at (v19) {$c^2_2$};
  \node[left-agent] at (v17) {$c^2_2$};
  \node[agent,line-c8,text=white] at (v37) {$c^2_3$};
  \node[left-agent] at (v27) {$c^2_3$};
  \node[agent,line-c6,text=white] at (v33) {$c^3_1$};
  \node[left-agent] at (v25) {$c^3_1$};

  \node[agent,line-c5,text=white] at (v24) {$c^2_1$};
  \node[agent,line-c2,text=white] at (v12) {$c^3_2$};
  \node[agent,line-c7,text=white] at (v26) {$c^1_3$};
  \node[agent,line-c9,text=white] at (v28) {$c^3_3$};
}
\newcommand{\steptwo}{
  \node[agent,line-c4,text=white] at (v34) {$c^1_1$};
  \node[left-agent] at (v23) {$c^1_1$};
  \node[agent,line-c1,text=white] at (v4) {$c^1_2$};
  \node[left-agent] at (v11) {$c^1_2$};
  \node[agent,line-c3,text=white] at (v5) {$c^2_2$};
  \node[left-agent] at (v17) {$c^2_2$};
  \node[agent,line-c8,text=white] at (v37) {$c^2_3$};
  \node[left-agent] at (v27) {$c^2_3$};
  \node[agent,line-c6,text=white] at (v6) {$c^3_1$};
  \node[left-agent] at (v25) {$c^3_1$};
}
\newcommand{\stepthree}{
  \node[agent,line-c5,text=white] at (v1) {$c^2_1$};
  \node[left-agent] at (v24) {$c^2_1$};
  \node[agent,line-c2,text=white] at (v2) {$c^3_2$};
  \node[left-agent] at (v12) {$c^3_2$};
  \node[agent,line-c7,text=white] at (v3) {$c^1_3$};
  \node[left-agent] at (v26) {$c^1_3$};
  \node[agent,line-c9,text=white] at (v10) {$c^3_3$};
  \node[left-agent] at (v28) {$c^3_3$};
}

\begin{figure*}[tb!]
  \centering
  \begin{tabular}{c}
    \begin{minipage}{1\linewidth}
      \centering
      \begin{tikzpicture}
        \body
        \stepone
      \end{tikzpicture}\\
      Step~1: Move assigned agents to vacation vertices
    \end{minipage}
    \bigskip\\
    \begin{minipage}{1\linewidth}
      \centering
      \begin{tikzpicture}
        \body
        \steptwo
        \node[agent,line-c5,text=white] at (v24) {$c^2_1$};
        \node[agent,line-c2,text=white] at (v12) {$c^3_2$};
        \node[agent,line-c7,text=white] at (v26) {$c^1_3$};
        \node[agent,line-c9,text=white] at (v28) {$c^3_3$};
      \end{tikzpicture}\\
      Step~2: Fill clause constrainers
    \end{minipage}
    \bigskip\\
    \begin{minipage}{1\linewidth}
      \centering
      \begin{tikzpicture}
        \body
        \draw[line-c5,ultra thick] (v1) to[out=-60,in=-120] (v2);
        \draw[line-c2,ultra thick] (v2) to[out=-60,in=-120] (v3);
        \draw[line-c7,ultra thick] (v3) to[out=-60,in=180] (v10);
        \draw[line-c9,ultra thick] (v10) to[out=0,in=-120] (v4);
        \draw[line-c1,ultra thick] (v4) -- (v5);
        \draw[line-c3,ultra thick] (v5) -- (v6);
        \draw[line-c6,ultra thick] (v6) to[out=60,in=120] (v7);
        \draw[line-c10,ultra thick] (circle-c1.east) -- (v1);
        \draw[line-c10,ultra thick] (v7) -- (circle-c2.west);
        \steptwo
        \stepthree
      \end{tikzpicture}\\
      Step~3: Move unassigned agents one step
    \end{minipage}
  \end{tabular}
  \caption{
    Construction of a reachable deadlock.
    The formula is $(x_1 \lor x_2 \lor \lnot x_3) \land (\lnot x_1 \lor x_2 \lor x_3) \land (x_1 \lor \lnot x_2 \land \lnot x3)$.
    The assignment is $x_1=\true$, $x_2=\true$, and $x_3=\true$.
    Locations of all agents are colored.
    When an agent departs from its start, the corresponding vertex is filled by dark.
    Bold lines in Step~3 constitute a reachable deadlock.
  }
  \label{fig:3-sat-deadlocks-steps}
\end{figure*}
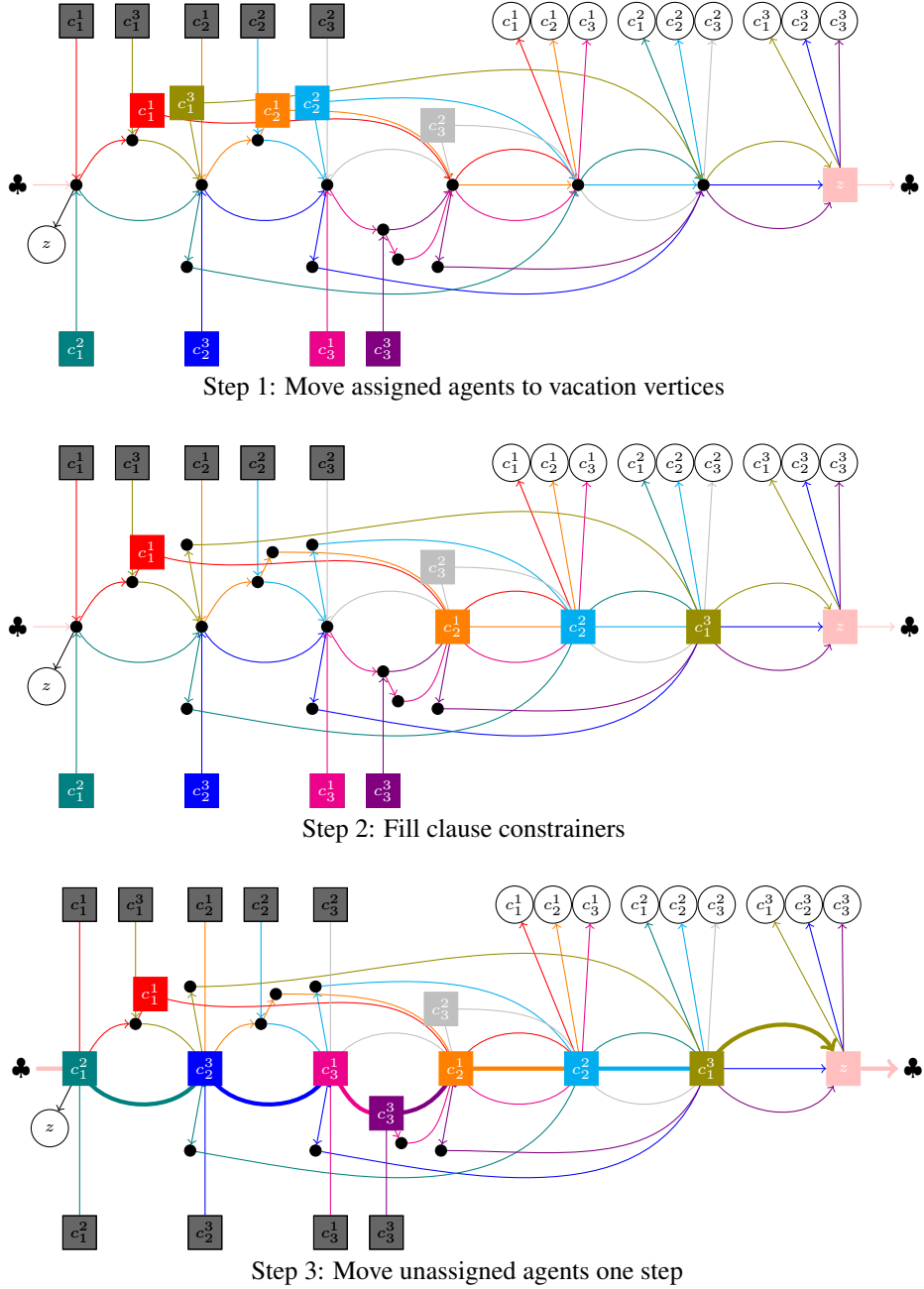
}

When the formula is satisfiable for one assignment, consider the following execution.
\begin{enumerate}
\setlength{\itemsep}{0pt}
\item For each assigned value, move the corresponding clause agents to vacation vertices in each variable decider, i.e., one step before clause constrainers.
\item Among the above agents, for each clause constrainer, there is at least one agent that can enter the clause constrainer due to satisfiability.
  Move them one step further.
  As a result, all clause constrainers have one agent at the first vertices.
  Vertices in upper/lower paths in the variable deciders must be vacant now.
\item Move all unassigned clause agents one step.
  As a result, all vertices in the unassigned paths are filled by the unassigned clause agents.
\end{enumerate}
\noindent
We now have a cyclic deadlock, i.e., this deadlock is reachable thus potential.

As an example, consider a satisfiable assignment $x_1=\true$, $x_2=\true$, $x_3=\true$.
In the beginning, move assigned agents, $c^1_1$, $c^1_2$, $c^2_2$, $c^2_3$, and $c^3_1$ to vacation vertices in each variable decider (Fig.~\ref{fig:3-sat-deadlocks-steps}; Step 1).
Next, move $c^1_2$, $c^2_2$, and $c^3_1$ to the first vertices of each clause constrainer of $C^1$, $C^2$, and $C^3$, respectively (Fig.~\ref{fig:3-sat-deadlocks-steps}; Step 2).
Then, move all unassigned agents, $c^2_1$, $c^3_2$, $c^1_3$, and $c^3_3$, one step (Fig.~\ref{fig:3-sat-deadlocks-steps}; Step 3).
There is a cyclic deadlock with $c^2_1, c^3_2, c^1_3, c^3_3, c^1_2, c^2_2, c^3_1$, and $z$, annotated with bold lines in Fig.~\ref{fig:3-sat-deadlocks-steps}.

\medskip
\noindent
\emph{C. The formula is satisfiable if a potential cyclic deadlock exists.}
To form a potential cyclic deadlock, for each variable decider, one or several agents try to move along either an upper or a lower path.
Consider assigning an opposite value against the used path to the variable.
For instance, if $c^1_2$ and $c^2_2$ involve in the deadlock at the variable decider (see Fig.~\ref{fig:3-sat-deadlocks}), then assign \false to $x_2$.
This assignment must satisfy the formula because at least one literal in each clause becomes true; otherwise, at least one clause constrainer exists such that the first vertex is empty, i.e., no deadlock.

\medskip
\noindent
\emph{D. All potential cyclic deadlocks are reachable.}
So far, we established the claim that a potential cyclic deadlock exists if and only if the formula is satisfiable.
Next, we claim that all potential cyclic deadlocks are reachable.
According to the above discussion, given a potential cyclic deadlock, the corresponding satisfiable assignment exists.
Consider the execution of Part~B using this assignment, slightly changing Step~2.
In this step, we can choose arbitrary agents for each clause constrainer.
Therefore, it is possible to choose agents involved in the potential cyclic deadlock.
As a result, this deadlock is reachable.
\end{proof}

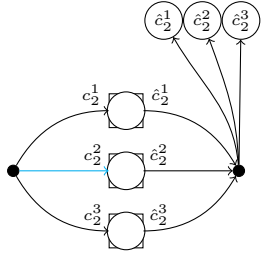
\begin{figure}[tb!]
  \centering
  \begin{tikzpicture}
    \newcommand{\edgesize}{1.6cm}
    \newcommand{\hedgex}{0.75}
    \newcommand{\hedgey}{0.6}
    \tikzset{
      sat-node/.style={vertex, minimum size=0.15cm, fill=black, anchor=center},
      start-node/.style={vertex, minimum size=0.45cm, text width=0.45cm, rectangle},
      goal-node/.style={vertex, minimum size=0.5cm},
      bold-line/.style={very thick},
      line-c3/.style={cyan,->},
    }
    \tiny
    \node[sat-node](v1) at (0, 0) {};
    \node[sat-node](v2) at (3, 0) {};
    \node[start-node](v3) at ($(v1) + (1.5,  0.8)$) {};
    \node[start-node](v4) at ($(v1) + (1.5,  0.0)$) {};
    \node[start-node](v5) at ($(v1) + (1.5, -0.8)$) {};
    \node[goal-node](v6) at (v3) {};
    \node[goal-node](v7) at (v4) {};
    \node[goal-node](v8) at (v5) {};
    \node[anchor=south east] at (v6.west) {$c^1_2$};
    \node[anchor=south west] at (v6.east) {$\hat{c}^1_2$};
    \node[anchor=south east] at (v4.west) {$c^2_2$};
    \node[anchor=south west] at (v4.east) {$\hat{c}^2_2$};
    \node[anchor=south east] at (v5.west) {$c^3_2$};
    \node[anchor=south west] at (v5.east) {$\hat{c}^3_2$};
    \node[goal-node](v9) at ($(v2) + (-1.0, 2.0)$) {$\hat{c}^1_2$};
    \node[goal-node, right=0 cm of v9](v10) {$\hat{c}^2_2$};
    \node[goal-node, right=0 cm of v10](v11) {$\hat{c}^3_2$};
    \draw[line,->] (v1) to[out=60,in=180] (v3);
    \draw[line-c3] (v1) to[out=0,in=180] (v4);
    \draw[line,->] (v1) to[out=-60,in=180] (v5);
    \draw[line,->] (v3) to[out=0,in=120] (v2);
    \draw[line,->] (v4) to[out=0,in=180] (v2);
    \draw[line,->] (v5) to[out=0,in=-120] (v2);
    \draw[line,->] (v2) to[out=95,in=-60] (v9);
    \draw[line,->] (v2) to[out=92,in=-70] (v10);
    \draw[line,->] (v2) -- (v11);
  \end{tikzpicture}
  \caption{Example of \textit{clause constrainer} without multiple edges.
   Used in the proof of Lemma~\ref{lemma:deadlock-np-comp}.
  }
  \label{fig:3-sat-multiple-edges}
\end{figure}

In the proof of Lemma~\ref{lemma:deadlock-np-comp}, we used multiple edges in a gadget \textit{clause constrainer} for the reduction from 3-SAT.
Since OTIMAPP assumes a simple graph (i.e., no multiple edges), we complement how to convert it to a \emph{correct} OTIMAPP instance.
Figure~\ref{fig:3-sat-multiple-edges} shows an example of the clause constrainer for $C^2$.
Recall that a clause constrainer contains all goals for the corresponding clause agents.
In this new gadget, we add intermediate vertices for each edge that potentially leads to cyclic deadlocks.
For each agent $c^j_k$, a new agent $\hat{c}^j_k$ is introduced.
Its start is the intermediate vertex.
Its goal is the original goal of $c^j_k$.
We furthermore change a goal for $c^j_k$ to starts of $\hat{c}^j_k$.
Consider now replacing all old clause constrainers with this new gadgets.
The translation is done in polynomial time.
The rest of the proof is straightforward from Lemma~\ref{lemma:deadlock-np-comp}.

\section{Detecting Potential Cyclic Deadlocks}
\label{sec:detection-details}

{
  \begin{algorithm}[tb!]
    \caption{Potential Cyclic Deadlock Detection}
    \label{algo:detecting-deadlock}
    \textbf{Input}:~a set of paths $\{ \path{1}, \ldots, \path{n} \}$\\
    \textbf{Output}:~one potential cyclic deadlock or \textbf{NONE}\\
    \small
    \begin{algorithmic}[1]
      \newcommand{\cfrom}{\m{\theta_{\mathit{f}}}}
      \newcommand{\cto}{\m{\theta_{\mathit{t}}}}
      \newcommand{\cnew}{\m{\theta}}
      \State $\tablefrom, \tableto \leftarrow \emptyset$
      \Comment table for fragments, key: vertex
      \medskip
      \Procedure{register}{$\theta$}
      \State $i, t_i, j, t_j \leftarrow \theta.\{\text{agents}, \text{indexes}\}\{[1], [-1]\}$
      \State $\tablefrom[\loc{i}{t_i}]$.append($\theta$)
      \State $\tableto[\loc{j}{t_j+1}]$.append($\theta$)
      \EndProcedure
      \medskip
      \Function{isPotentialDeadlock}{$\theta$}
      \State $i, t_i, j, t_j \leftarrow \theta.\{\text{agents}, \text{indexes}\}\{[1], [-1]\}$
      \State \Return $\loc{i}{t_i} = \loc{j}{t_j+1}$
      \EndFunction
      \medskip
      \For{$i = 1\ldots n$}
      \label{algo:deadlock:for-agents}

      \For{$j= 1\ldots |\path{i}|-1$}
      \label{algo:deadlock:for-t}
      \State $u \leftarrow \loc{i}{j}, v \leftarrow \loc{i}{j+1}$
      \medskip
      \State $\cnew \leftarrow \fragment{(i)}{(j)}$
      \label{algo:deadlock:add-own}
      \Comment{single fragment}
      \State \Call{register}{$\cnew$}
      \label{algo:deadlock:add-own:register}
      \medskip
      \For{$\cto \in \tableto[u]$}
      \label{algo:deadlock:tableto}
      \Comment{check fragments on \tableto}
      \IFSINGLE{$i \in \cto.\text{agents}$}{\CONTINUE}
      \State $\cnew \leftarrow \fragment{\cto.\text{agents} + i}{\cto.\text{indexes} + j}$
      \IFSINGLE{\Call{isPotentialDeadlock}{$\theta$}}{\textbf{return} \cnew}
      \label{algo:deadlock:tableto:detecting}
      \State \Call{register}{$\cnew$}
      \EndFor
      \label{algo:deadlock:tableto:end}
      \medskip
      \For{$\cfrom \in \tablefrom[v]$}
      \label{algo:deadlock:tablefrom}
      \Comment{check fragments on \tablefrom}
      \IFSINGLE{$i \in \cfrom.\text{agents}$}{\CONTINUE}
      \State $\cnew \leftarrow \begin{aligned}[t] \{
        \text{agents:}~i + \cfrom.\text{agents},
        \text{indexes:}~j + \cfrom.\text{indexes}
      \}\end{aligned}$
      \IFSINGLE{\Call{isPotentialDeadlock}{$\theta$}}{\textbf{return} \cnew}
      \label{algo:deadlock:tablefrom:detecting}
      \State \Call{register}{$\cnew$}
      \EndFor
      \label{algo:deadlock:tablefrom:end}
      \medskip
      \For{$\cto \in \tableto[u]$, $\cfrom \in \tablefrom[v]$}
      \label{algo:deadlock:connect-two}
      \Comment{connect two fragments}
      \IFSINGLE{$i \in \cto.\text{agents} \cup \cfrom.\text{agents}$}{\CONTINUE}
      \IFSINGLE{$\cto.\text{agents} \cap \cfrom.\text{agents} \neq \emptyset$}{\CONTINUE}
      \State $\cnew \leftarrow \begin{aligned}[t] \{
        &\text{agents:}~\cto.\text{agents} + i + \cfrom.\text{agents}\\
        &\text{indexes:}~\cto.\text{indexes} + j + \cfrom.\text{indexes}
        \}\end{aligned}$
      \IFSINGLE{\Call{isPotentialDeadlock}{$\theta$}}{\textbf{return} \cnew}
      \label{algo:deadlock:connect-two:detecting}
      \State \Call{register}{$\cnew$}
      \EndFor
      \label{algo:deadlock:connect-two:end}
      \EndFor
      \label{algo:deadlock:for-t-end}
      \EndFor
      \label{algo:deadlock:for-agents-end}
      \State \Return \textbf{NONE}
    \end{algorithmic}
  \end{algorithm}
}

Using fragments, Alg.~\ref{algo:detecting-deadlock} detects a potential cyclic deadlock in a set of paths if exists.
The intuition is the following:
(1)~the algorithm checks each path one by one,
(2)~it stores all fragments created so far,
(3)~for each edge in each path, it creates new fragments using existing fragments, and
(4)~if a fragment ends at its start, this is a potential cyclic deadlock.
We describe the details in the proof of the completeness.

\begin{theorem}[completeness]
  Alg.~\ref{algo:detecting-deadlock} finds one potential cyclic deadlock if exists, otherwise returns \textbf{NONE}.
\end{theorem}
\begin{proof}
  The algorithm uses two tables that store fragments: \tablefrom and \tableto.
  Both tables take one vertex as a key.
  One entry in \tablefrom stores all fragments starting from the vertex.
  One entry in \tableto stores all fragments ending at the vertex.
  A fragment is registered in both tables.
  We now derive the statement by induction on \path{i}.

  \medskip
  \noindent
  \emph{Base case}:
  At the first iteration of the loop [Line~\ref{algo:deadlock:for-agents}--\ref{algo:deadlock:for-agents-end}], all fragments for $\{ \path{1} \}$ are registered on \tablefrom and \tableto due to Line~\ref{algo:deadlock:add-own}--\ref{algo:deadlock:add-own:register}.
  There are no potential cyclic deadlocks for $\{ \path{1} \}$.

  \medskip
  \noindent
  \emph{Induction Hypothesis}:
  Assume that there are no potential cyclic deadlocks for $\{ \path{1}, \ldots, \path{i-1} \}$ and all fragments for them are registered on \tablefrom and \tableto.

  \medskip
  \noindent
  \emph{Induction Step}:
  We now show the property for $i$; otherwise, a potential cyclic deadlock exists for $\{ \path{1}, \ldots, \path{i} \}$ and the algorithm returns it.
  All new fragments about \path{i} are categorized into two:
  (1)~a fragment only with \path{i} or
  (2)~a fragment that extends other fragments on \tablefrom and \tableto, using $(u, v) \in \path{i}$.
  The former is preserved due to Line~\ref{algo:deadlock:add-own}--\ref{algo:deadlock:add-own:register}.
  The latter is further categorized into three cases:
  (a)~a fragment ends at $v$,
  (b)~a fragment starts from $u$, and
  (c)~a fragment connecting two existing fragments that one ends at $u$ and another starts from $v$.
  Each case corresponds to Line~\ref{algo:deadlock:tableto}--\ref{algo:deadlock:tableto:end}, Line~\ref{algo:deadlock:tablefrom}--\ref{algo:deadlock:tablefrom:end}, and Line~\ref{algo:deadlock:connect-two}--\ref{algo:deadlock:connect-two:end}, respectively.
  As a result, all fragments are to register on \tablefrom and \tableto;
  otherwise, a potential cyclic deadlock exists and the algorithm returns it [Line~\ref{algo:deadlock:tableto:detecting}, \ref{algo:deadlock:tablefrom:detecting}, and \ref{algo:deadlock:connect-two:detecting}].
\end{proof}

The time complexity does not contradict the NP-completeness of detecting potential deadlocks (Lemma~\ref{lemma:deadlock-np-comp}).
\begin{proposition}[space and time complexity]
  Algorithm~\ref{algo:detecting-deadlock} requires $\Omega(2^{|n|})$ both for space and time complexity in the worst case.
\end{proposition}
\begin{proof}
  Consider an example in Fig.~\ref{fig:intractable-example}.
  In any solutions, the number of fragments starting from $u$ becomes $\Omega\left(2^{|n|}\right)$; this implies the statement.
\end{proof}

Although Alg.~\ref{algo:detecting-deadlock} does not run in polynomial time, it works sufficiently fast in a sparse environment such that not many paths use the same vertices.

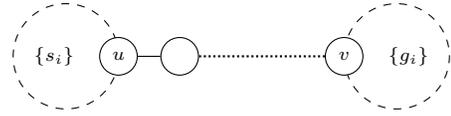
\begin{figure}[tb!]
  \centering
  \begin{tikzpicture}
    \newcommand{\edgesize}{0.2cm}
    \scriptsize
    \coordinate[](loc1) at (0, 0);
    \coordinate[](loc3) at (3.0, 0);
    \draw[line, dashed] (loc1) arc(0:360:0.7);
    \draw[line, dashed] (loc3) arc(0:360:-0.7);
    \node[vertex,fill=white](v1) at (loc1) {$u$};
    \node[vertex, right=0.3cm of v1](v2) {};
    \node[vertex,fill=white](v3) at (loc3) {$v$};
    \draw[line] (v1) -- (v2);
    \draw[line, densely dotted, thick] (v2) -- (v3);
    \node[left=0.5cm of loc1,anchor=east]() {$\{ s_i \}$};
    \node[right=0.5cm of loc3,anchor=west]() {$\{ g_i \}$};
  \end{tikzpicture}
  \caption{
    Example that requires huge space and time to detect potential deadlocks.
    All starts are on the left. All goals are on the right.
    Two zones are connected by a sufficiently long path.
  }
  \label{fig:intractable-example}
\end{figure}

\section{Proof of DBS}
\label{sec:proof:dbs}
\begin{theorem*}[\ref{thrm:dbs}; DBS]
  DBS returns a solution when solutions satisfying Thm.~\ref{thrm:sufficient} exist; otherwise returns \FAILURE.
\end{theorem*}
\begin{proof}
  Assume that there is a solution $\paths = \{ \path{1}, \ldots, \path{N} \}$ satisfying the relaxed sufficient condition (Thm.~\ref{thrm:sufficient}).
  At each cycle [Line~\ref{algo:cp:while}--\ref{algo:cp:while:end}], at least one node in \open is \emph{consistent} with \paths, i.e., its constraints allow searching \paths.
  This is derived by induction:
  (1) the initial node $R$ is consistent with \paths, and
  (2) generated nodes from a consistent node with \paths must include at least one consistent node.
  The search space, i.e., which agents are prohibited using which edges in which directions, is finite.
  Therefore, DBS eventually returns \paths (or another solution); otherwise, such solutions do not exist.
\end{proof}

\section{Stress Test on Random Graphs}
\label{sec:stress-test-random}

{
  \setlength{\tabcolsep}{0pt}
  \newcommand{\colwidth}{0.32\hsize}
  \newcommand{\plotStress}[1]{
    \begin{minipage}[t]{\colwidth}
      \centering
      \IfFileExists{fig/raw/stress-#1.pdf}{\includegraphics[width=1.0\linewidth]{fig/raw/stress-#1.pdf}}{}
    \end{minipage}
  }
  \newcommand{\plotCactus}[1]{
    \begin{minipage}[t]{\colwidth}
      \centering
      \IfFileExists{fig/raw/cactus-#1.pdf}{\includegraphics[width=1.0\linewidth]{fig/raw/cactus-#1.pdf}}{}
    \end{minipage}
  }
  \newcommand{\plotProfiling}[1]{
    \begin{minipage}[t]{\colwidth}
      \centering
      \IfFileExists{fig/raw/profiling-#1.pdf}
      {\includegraphics[width=0.95\linewidth,right]{fig/raw/profiling-#1.pdf}}{}
    \end{minipage}
  }
  \newcommand{\headerER}[2]{
    \begin{minipage}[t]{\colwidth}
      \centering{\scriptsize #2}
    \end{minipage}
  }
  \begin{figure}[tb!]
    \centering
    \begin{tabular}[t]{ccc}
      \headerER{ER-3}{$G(1000, 0.003)$}
      & \headerER{ER-4}{$G(1000, 0.004)$}
      & \headerER{ER-5}{$G(1000, 0.005)$}
      \\
      \plotStress{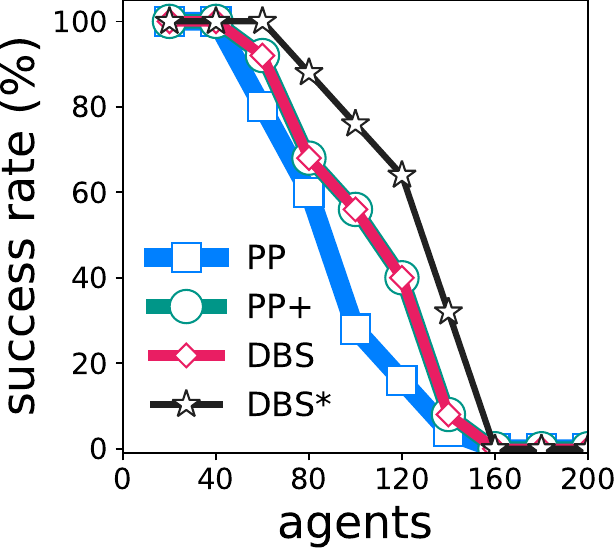}
      & \plotStress{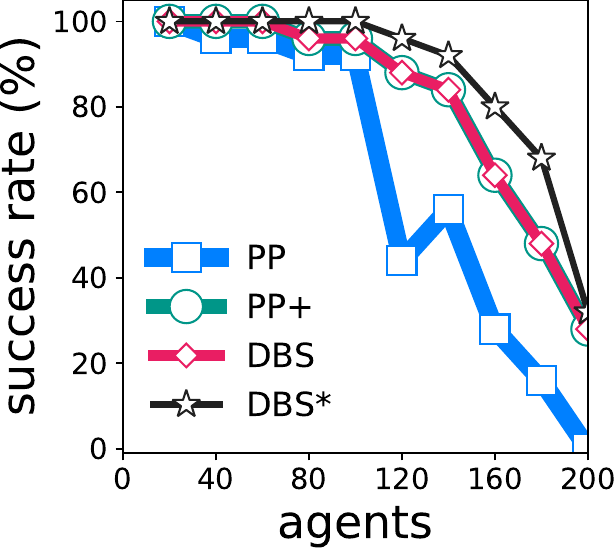}
      & \plotStress{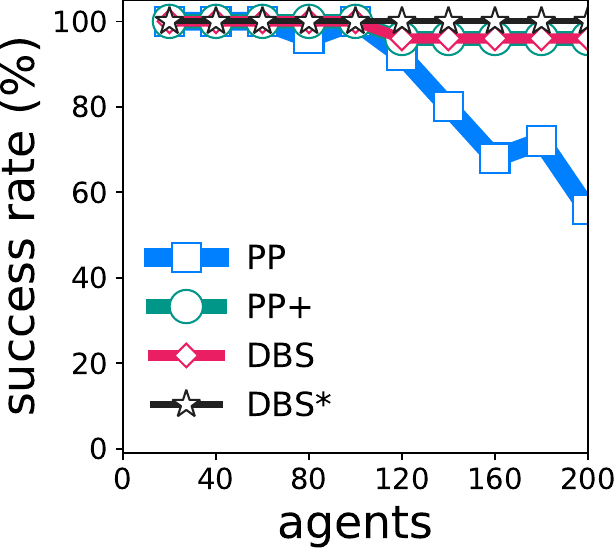}
    \end{tabular}
    \caption{
      Stress test of random graphs.
      The success rate is based on 25 identical instances.
      DBS$^\ast$ includes detected instances that are unsolvable for DBS before timeout, which is not possible for PP$^{(+)}$.
      Results on random graphs $G(n,p)$ are shown, where $n$ is the number of vertices and every possible edge occurs independently with probability $p$.
    }
    \label{fig:stress-test-random}
  \end{figure}
}

Figure~\ref{fig:stress-test-random} summarizes the results.
The experimental setting is the same as Sec.~\ref{sec:stress-test}.
We can see that the difficulty of finding solutions is dominated by average degrees of graphs.

\section{Details of Experimental Setup}
\label{sec:exp-detail}
\subsection{Implementation of DBS}
An initial solution candidate is important for DBS.
It is ideal to find solutions (i.e., a set of paths without deadlocks) from the beginning.
Even if not, it is desired to obtain infeasible solutions with a small number of potential cyclic deadlocks, expected to expand a smaller number of nodes in the high-level search to reach feasible solutions.
We thus made the low-level search for the initial solution take a path having fewer potential cyclic deadlocks with already planned paths, partially using Alg.~\ref{algo:detecting-deadlock}.
This is akin to tie-breaks in low-level search of CBS~\cite{sharon2015conflict}.

\subsection{Setup of MAPF-DP}
We carefully designed experiments to be fair as follows.

\paragraph{Preliminaries}
MAPF-DP~\cite{ma2017multi} emulates the imperfect execution of MAPF plans by introducing the possibility of unsuccessful moves, but still agents have to avoid collisions.
Time is discrete.
At each timestep, an agent $i$ can either stay in place or move to an adjacent vertex with a probability $p_i$ of being unsuccessful.
Solution quality is assessed by the total traveling time, where the time is the earliest time step that one agent reaches its goal and remains there.

\paragraph{From OTIMAPP to MAPF-DP}
To adapt the execution of OTIMAPP to MAPF-DP, we introduce two changes for executions:
(1)~using \emph{mode}s to represent a state on edges, and,
(2)~an activation rule to represent the failure of movements.

\medskip
\noindent
\emph{Mode}:
In reality, an agent $i$ occupies two vertices simultaneously during a move from one vertex to another vertex.
We introduce two \emph{modes} in the execution of OTIMAPP to represent this state;
\begin{itemize}
  \setlength{\itemsep}{0pt}
  \item A mode \contracted corresponds to when the agent $i$ occupies one vertex.
  \item A mode \extended corresponds to when the agent $i$ occupies two vertices.
\end{itemize}
\noindent
Agents move towards their goals by changing two modes alternately.
Initially, they are in \contracted.
The names are from~\cite{okumura2021time}.

\medskip
\noindent
\emph{Activation Rule}:
We repeated the following two phases:
\begin{itemize}
  \setlength{\itemsep}{0pt}
\item Each agent $i$ in \extended is activated with probability $1-p_i$.
  As a result, the agent $i$ successfully moves to the adjacent vertex with probability $1-p_i$ and becomes \contracted.
\item Choose one agent in \contracted randomly then makes it activated.
  Repeat this until the configuration becomes stable, i.e, all agents in \contracted do not change their states unless any agent in \extended is activated.
\end{itemize}
\noindent
A pair of the two phases is regarded as one timestep.

\paragraph{Other Experimental Setup}
The delay probabilities $p_i$ were chosen uniformly at random from $[0, \bar{p}]$, where $\bar{p}$ is the upper bound of $p_i$.
The higher $\bar{p}$ means that agents delay often.
$\bar{p}=0$ corresponds to perfect executions without delays.
Implementations of the online time-independent planning, called Causal-PIBT, were obtained from the authors~\cite{okumura2021time}.
The offline MAPF plans for MCPs~\cite{ma2017multi} was obtained by ECBS~\cite{barer2014suboptimal}, a bounded sub-optimal solver for MAPF.
The sub-optimally was set to $1.1$, which was adjusted to solve all instances in the experiment.
The implementation of ECBS was obtained from~\citeAppendix{okumura2021iterative} (in the additional references).

\subsection{Setup of Robot Demonstrations}
\subsubsection{Centralized Execution}
\paragraph{Platform}
We used the \emph{toio} robots (\url{https://toio.io/}).
The toio robots, connected to a computer via BLE~(Bluetooth Low Energy), evolve on a specific playmat and are controllable by instructions of absolute coordinates.
We informally confirmed that there is a non-negligible action delay between robots when sending instructions to several robots simultaneously (e.g., 10 robots, see the movie).
Therefore, one-shot execution --- robots move alone without communication after the receipt of plans --- will result in collisions hence failure of the execution in a high possibility.
The robots need some kinds of execution policies.

\paragraph{Usage}
We created a virtual grid on the playmat and the robots followed the grid.
A central server (a laptop) managed the locations of all robots and issued the instructions (i.e., where to go) to each robot step by step.
The instructions were issued \emph{asynchrony} between robots while avoiding collisions.
The code was written in Node.js.

\subsubsection{Decentralized Execution}
\paragraph{Platform}
We used the \emph{AFADA} platform~\cite{kameyama2021active};
an architecture that consists of mobile robots that evolve over an active environment made of flat \emph{cells} each equipped with a computing unit (Fig.~\ref{fig:afada-description}).
Adjacent cells can communicate with each other via a serial interface.
Cells form the environment in two ways:
as a two-dimensional physical grid, and, as a communication network.
In addition, a cell can communicate with robots on it via NFC (Near Field Communication).

{
  \begin{figure}[t]
    \centering
    \includegraphics[width=0.5\hsize]{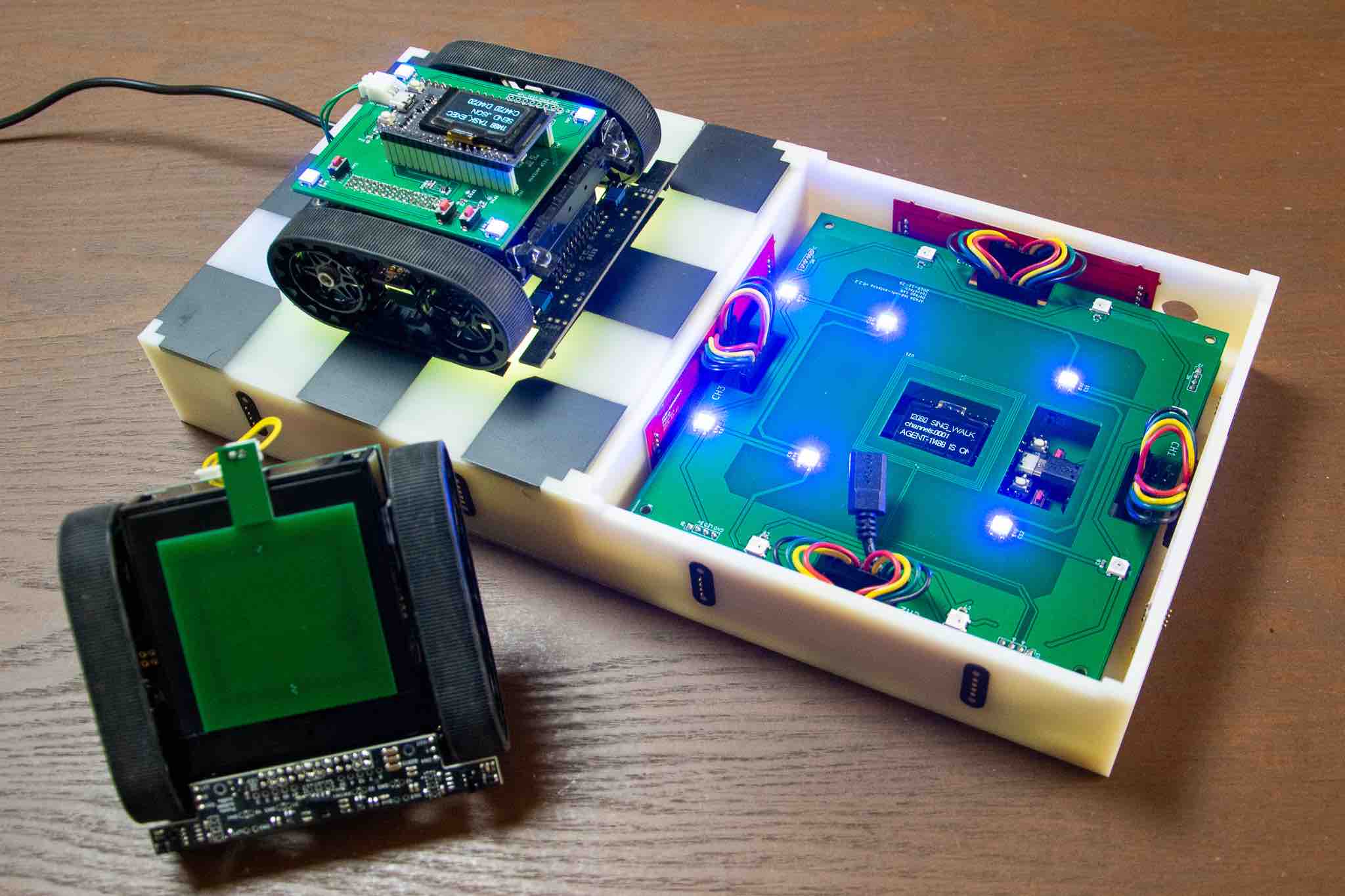}
    \caption{
      Components of AFADA.
      The system consists of two kinds of actors: cells and robots.
      A cell is covered by an acrylic roof patterned for line tracing (left cell), enabling robots to move following the grid structure.
    }
    \label{fig:afada-description}
  \end{figure}
}

\paragraph{Usage}
Robots first receive the OTIMAPP solution from a laptop via Wi-Fi, then move following the plan.
Cells achieve mutual exclusion of locations for robots, i.e., collision avoidance, using local communication as follows.
Before moving to the next vertex (i.e., cell; denoted as \vnext), a robot first asks the underlying cell \vcurrent the availability of \vnext.
Then, \vcurrent asks \vnext its status.
If \vnext is \emph{reserved} by another robot, \vcurrent waits a while and asks the status of \vnext again; otherwise, \vcurrent makes \vnext reserved and notifies the robot to move to \vnext.
When the robot reaches \vnext, then the robot \emph{releases} \vcurrent via \vnext.
Importantly, there is \emph{no central control} at runtime.
Any actor (robots, cells, and the laptop that sends the plan) has no methods to know the entire configuration.
This also means that the system is fully asynchronous as for timing.
Furthermore, there is no global communication; robots and cells decide their actions based on information from nearby actors.

\bibliographystyleAppendix{named}
\bibliographyAppendix{ref}

\end{document}